\documentclass[12pt,a4paper]{article}

\usepackage[margin=1in]{geometry}

\usepackage[english,activeacute]{babel}
\usepackage{etex}
\usepackage[latin1]{inputenc}
\usepackage[all]{xy}
\usepackage[latin1]{inputenc}
\usepackage{xypic,color}
\usepackage{latexsym}
\usepackage{mathtools}
\usepackage{amsmath}
\usepackage{amssymb}
\usepackage{wasysym}
\usepackage{mathrsfs}
\usepackage{euscript}
\usepackage{amsthm}
\usepackage{amsmath}
\usepackage{amsfonts}
\usepackage{cancel}
\usepackage{graphicx,color}
\usepackage{multicol}
\usepackage{color}
\usepackage{multirow}
\usepackage{makeidx}
\usepackage{upgreek}
\usepackage{setspace} 
\usepackage{import}
\usepackage{float}
\usepackage{mleftright}
\usepackage{stmaryrd}
\usepackage[table,xcdraw]{xcolor}
\restylefloat{table}
\usepackage{faktor}

\newcommand{\brbinom}[2]{\genfrac{[}{]}{0pt}{}{#1}{#2}}

\usepackage[T1]{fontenc}

\usepackage{array}
\newcommand{\PreserveBackslash}[1]{\let\temp=\\#1\let\\=\temp}
\newcolumntype{C}[1]{>{\PreserveBackslash\centering}p{#1}}
\newcolumntype{R}[1]{>{\PreserveBackslash\raggedleft}p{#1}}
\newcolumntype{L}[1]{>{\PreserveBackslash\raggedright}p{#1}}

\usepackage{bbold}

\usepackage{tikz}
\usepackage{tikz-cd}
\usetikzlibrary{matrix,arrows,calc,decorations.pathmorphing,fit,shapes.geometric,decorations.pathreplacing,positioning}

\usepackage{xcolor}
\definecolor{babyblueeyes}{rgb}{0.54, 0.81, 0.94}
\definecolor{blizzardblue}{rgb}{0.67, 0.9, 0.93}
\definecolor{blue(munsell)}{rgb}{0.0, 0.5, 0.69}
\definecolor{bluegray}{rgb}{0.4, 0.6, 0.8}
\definecolor{bondiblue}{rgb}{0.0, 0.58, 0.71}
\definecolor{cadetblue}{rgb}{0.37, 0.62, 0.63}
\definecolor{carolinablue}{rgb}{0.6, 0.73, 0.89}
\definecolor{cinnamon}{rgb}{0.82, 0.41, 0.12}
\definecolor{darkcandyapplered}{rgb}{0.64, 0.0, 0.0}
\definecolor{darkcyan}{rgb}{0.0, 0.55, 0.55}
\definecolor{darkmidnightblue}{rgb}{0.0, 0.2, 0.4}
\definecolor{darkpastelblue}{rgb}{0.47, 0.62, 0.8}
\definecolor{frenchblue}{rgb}{0.0, 0.45, 0.73}

\usepackage{url}
\usepackage[bookmarksopen,colorlinks,allcolors=frenchblue,urlcolor=frenchblue]{hyperref}

\makeindex

\usepackage{scalerel}

\usepackage{wrapfig}
\usepackage{ifpdf}

\usepackage{syntonly}

\DeclareMathOperator{\ob}{\mathsf{ob}}
\DeclareMathOperator{\id}{\mathsf{id}}

\DeclareMathOperator{\im}{Im}

\DeclareMathOperator{\col}{\mathsf{col}}
\DeclareMathOperator{\source}{\mathsf{s}}
\DeclareMathOperator{\target}{\mathsf{t}}

\DeclareMathOperator{\upk}{\textup{k}}
\DeclareMathOperator{\upf}{\textup{f}}
\DeclareMathOperator{\upg}{\textup{g}}
\DeclareMathOperator{\uph}{\textup{h}}

\DeclareMathOperator{\upj}{\textup{j}}
\DeclareMathOperator{\upn}{\textup{n}}

\DeclareMathOperator{\upr}{\textup{r}}
\DeclareMathOperator{\upb}{\textup{b}}
\DeclareMathOperator{\upo}{\textup{o}}

\DeclareMathOperator{\upu}{\textup{u}}
\DeclareMathOperator{\upv}{\textup{v}}

\DeclareMathOperator{\upd}{\textup{d}}

\DeclareMathOperator{\Aalg}{\mathsf{A}}
\DeclareMathOperator{\Balg}{\mathsf{B}}

\DeclareMathOperator{\Set}{\mathsf{Set}}
\DeclareMathOperator{\sSet}{\mathsf{Set}_{\Delta}}

\DeclareMathOperator{\Vrect}{\mathsf{V}}
\DeclareMathOperator{\Wrect}{\mathsf{W}}
\DeclareMathOperator{\Urect}{\mathsf{U}}

\DeclareMathOperator{\Op}{\EuScript{O}}

\DeclareMathOperator{\OpP}{\EuScript{P}}
\DeclareMathOperator{\OpN}{\EuScript{N}}
\DeclareMathOperator{\OpB}{\EuScript{B}}

\DeclareMathOperator{\M}{\EuScript{M}}

\DeclareMathOperator{\V}{\EuScript{V}}

\DeclareMathOperator{\Frect}{\mathsf{F}}
\DeclareMathOperator{\Grect}{\mathsf{G}}
\DeclareMathOperator{\Mrect}{\mathsf{M}}
\DeclareMathOperator{\Drect}{\mathsf{D}}

\DeclareMathOperator{\Crect}{\mathsf{C}}

\DeclareMathOperator{\Lrect}{\mathsf{L}}

\DeclareMathOperator{\Qrep}{\mathfrak{Q}}

\DeclareMathOperator{\Srect}{\mathsf{S}}

\DeclareMathOperator{\X}{\mathsf{X}}
\DeclareMathOperator{\Xrect}{\mathsf{X}}
\DeclareMathOperator{\Y}{\mathsf{Y}}

\DeclareMathOperator{\T}{\mathsf{T}}

\DeclareMathOperator{\Q}{\mathsf{Q}}
\DeclareMathOperator{\q}{\mathsf{q}}
\DeclareMathOperator{\Cof}{\mathsf{Cof}}
\DeclareMathOperator{\Fib}{\mathsf{Fib}}
\DeclareMathOperator{\QFT}{\mathsf{QFT}}

\DeclareMathOperator{\Mon}{\mathsf{Mon}}
\DeclareMathOperator{\Map}{\mathsf{Map}}
\DeclareMathOperator{\Qft}{\mathfrak{Qft}}

\DeclareMathOperator{\Disc}{\mathsf{Disc}}

\DeclareMathOperator{\Ho}{\mathsf{Ho}}
\DeclareMathOperator{\Fun}{\mathsf{Fun}}

\DeclareMathOperator{\h}{\mathsf{h}}

\DeclareMathOperator{\Ch}{\mathsf{Ch}}

\DeclareMathOperator{\pt}{\mathrm{pt}}

\DeclareMathOperator{\Alg}{\mathsf{Alg}}

\usepackage{titlesec}

\usepackage{float}

\usepackage{fancyhdr}
\usepackage{tikz}
\usetikzlibrary{calc,decorations.pathmorphing,shapes}

\newcounter{sarrow}

 \newtheorem{thm}{Theorem}[section]
 \newtheorem{taggedtheoremx}{Theorem}
 \newenvironment{taggedtheorem}[1]
 {\renewcommand\thetaggedtheoremx{#1}\taggedtheoremx}
 {\endtaggedtheoremx}
 \newtheorem{cor}[thm]{Corollary}
 \newtheorem{taggedcorx}{Corollary}
 \newenvironment{taggedcor}[1]
 {\renewcommand\thetaggedcorx{#1}\taggedcorx}
 {\endtaggedcorx}
 \newtheorem{lem}[thm]{Lemma}
 \newtheorem{prop}[thm]{Proposition}
 \newtheorem{taggedpropx}{Proposition}
 \newenvironment{taggedprop}[1]
 {\renewcommand\thetaggedpropx{#1}\taggedpropx}
 {\endtaggedpropx}
 
 \newtheorem{hyp}[thm]{Hypothesis}
 \newtheorem{assump}[thm]{Assumption}
 
 \theoremstyle{definition}
 \newtheorem{defn}[thm]{Definition}
 \newtheorem{notat}[thm]{Notation}
 
 \theoremstyle{remark}
 \newtheorem{rem}[thm]{Remark}
 \newtheorem{rems}[thm]{Remarks}
 
 \newtheorem{ex}[thm]{Example}

\title{
\begin{Large}
\textsc{New model category structures for Algebraic Quantum Field Theory}	
\end{Large}
}
\author{Victor Carmona\thanks{The author was partially supported by the Spanish Ministry of Economy under the grant MTM2016-76453-C2-1-P (AEI/FEDER, UE), by the Andalusian Ministry of Economy and Knowledge and the Operational Program FEDER 2014-2020 under the grant US-1263032, by grant PID2020-117971GB-C21 of the Spanish Ministry of Science and Innovation, and grant FQM-213 of the Junta de Andaluc\'ia. He was also partly supported by Spanish Ministry of Science, Innovation and Universities grant FPU17/01871.}
}

\begin{document}

\maketitle
\vspace*{1mm}
\begin{abstract}
In this paper we define and compare several new Quillen model structures which present the homotopy theory of algebraic quantum field theories. In this way, we expand foundational work of Benini, Schenkel and collaborators \cite{benini_higher_2019} by providing a richer framework to detect and treat homotopical phenomena in quantum field theory. Our main technical tool is a new \emph{extension model} structure on operadic algebras which is constructed via (right) Bousfield localization.

This document also includes a correction to the published article in Appendix \ref{app_Corrigendum}. \vspace*{5mm}	
\end{abstract}

\tableofcontents

\begin{section}{Introduction}

The problem of applying quantum mechanical principles to classical systems of fields was the foundational idea which motivates the emersion of Quantum Field Theories. However, a general formulation of what is precisely a Quantum Field Theory seems to be rather elusive. This fact gave rise to the appearance of a broad variety of competing definitions to deal with them formally. A very rough classification of such candidates could be: perturbative models, which use the knowledge about simpler free models \cite{rejzner_perturbative_2016}; functorial approaches, which are centered on the study of spaces of fields and propagators \cite{lurie_classification_2009,stolz_supersymmetric_2011}; and algebraic formulations, which deal with the possible algebraic structures that the collection of observables may carry \cite{fewster_algebraic_2015}\footnote{Of course, this classification is far from sharp.}. In this work, we focus on the algebraic picture, i.e.\ on the formal study of algebras of observables, which is commonly known as Algebraic Quantum Field Theory and usually denoted AQFT. 

Currently, the denomination AQFT should be seen as a great umbrella which covers a variety of different examples (see Section \ref{sect_Prelim}). Nonetheless, let us briefly recall a relaxed Haag-Kastler axiomatization of AQFTs \cite{haag_algebraic_1964} to get a grasp of the objects under study. 
 Fix a Lorentzian manifold $\Mrect$ which should be interpreted as the background spacetime. Then, an AQFT over $\Mrect$ is given by a functorial assignment of local operator algebras to suitable open subregions of $\Mrect$, i.e. a rule $\Aalg$ that assigns   
\begin{itemize}
	\item an unital associative algebra $\Aalg(\Urect)$ to each suitable open subset $\Urect\subseteq\Mrect$ and 
	\item an algebra morphism $\Aalg(\Urect)\to\Aalg(\Vrect)$ to each inclusion $\Urect\subseteq\Vrect$ in  $\Mrect$
\end{itemize}
subject to functoriality constraints and which satisfies \emph{causality} and a \emph{time-slice axiom}. On the one hand, causality encapsulates the relativistic nature of observables; it roughly says that observables coming from ``spacelike separated" open regions commute. On the other hand, the time-slice axiom asserts that observables on a region only depend on an arbitrary small neighborhood of any of its Cauchy surfaces, forcing time evolution to be controlled by well-defined initial value problems. 

Such a simple axiom system has lead to several applications in physics (see \cite{fredenhagen_quantum_2007}). To name one of them strongly related with our work, Fredenhagen in \cite{fredenhagen_global_1993} was able to compute global non-perturbative phenomena over the circle by gluing in a precise way observables supported on intervals. Therefore, Fredenhagen's contribution deserves the name of \emph{local-to-global principle} and it addresses the important question of whether or not global observables are fully determined by local observables and how one can reconstruct global observables from local data. 

Our starting point comes from a recent development in the field of AQFTs which has been undertaken in the program of Benini, Schenkel and collaborators \cite{benini_linear_2020, benini_algebraic_2018,benini_homotopy_2015,benini_homotopy_2019, benini_operads_2021}. Their fundamental contribution was to identify operadic constructions, and their importance, on previous axiomatizations and foundations, e.g.\ providing a simple axiom system that covers a variety of examples (orthogonal categories) or justifying Fredenhagen's local-to-global principle \cite{fredenhagen_generalizations_1990} from a structured point of view, as well as detecting its weakness when causality comes into play. Reasonably, a good deal of new questions arises by the introduction of operad theory into the branch of AQFTs, as mentioned in the survey \cite{benini_higher_2019}. 

Operad theory \cite{markl_operads_2002} has been proven to be a fruitful mathematical machinery to deal with algebraic structures. Its importance in mathematics is doubtless due to its multidisciplinarity, as it appears in homotopy theory, algebraic topology, geometry, category theory, mathematical physics or mathematical programming. For instance, this theory has also been a valuable tool in Kontsevich's work on quantization of Poisson manifolds \cite{kontsevich_deformation_2003} or Costello-Gwilliam's factorization algebras in perturbative field theories \cite{costello_factorization_2017}. In the AQFT world, one of the main advantages of working with operads is the possibility of introducing manageable and computable homotopical techniques in the form of Quillen model structures, as observed in \cite{benini_homotopy_2019}. On the one hand, embracing homotopical mathematics permits a better understanding of non discrete phenomena such as gauge symmetries or the presence of anti-fields and anti-ghosts in the BV-BRST formalism (e.g. \cite{benini_quantum_2017, benini_homotopy_2015}). These facts motivate a deeper immersion into the homotopy theory of AQFTs. On the other hand, the presence of a Quillen model structure brings to our disposal an assorted toolkit to understand a particular homotopy theory.  

In \cite{benini_operads_2021, bruinsma_algebraic_2019}, the algebraic structure of an AQFT is encoded via a naturally defined operad, and so it is reasonable to study the homotopy theory of AQFTs using a canonical model structure on operadic algebras \cite{white_bousfield_2018} (see Theorem \ref{thm_ProjectiveModelOnPlainAQFTs} and the previous discussion). However, this approach is not well suited to introduce further axioms on AQFTs which are not structural. For instance, it does not take into account dynamic axioms or local-to-global principles properly. The present work exploits the operadic description of AQFTs one step further to find new model structures which fix these problems to study more advanced homotopy theories.

Let us highlight the main results achieved in this work and their motivation. 

In the sequel of this introduction, readers that are not familiar with orthogonal categories (Definition \ref{defn_OrthogonalCategory}) may interpret an orthogonal category $\Crect^{\perp}$ as the category of suitable open regions of a Lorentzian manifold $\Mrect$, together with the information about causal disjointness. Also, we will denote by $\QFT(\Crect^{\perp})$ the ordinary category of AQFTs over $\Crect^{\perp}$ (valued in chain-complexes over a field of characteristic 0) in the remainder of this introduction.

When dealing with the homotopy theory of AQFTs, a dynamical axiom such as \emph{time-slice} (Definition \ref{defn_QFTsSatisfyingTimeSliceAxiom}) seems to be excessively demanding as the example provided by \cite[Theorem 6.8]{benini_quantum_2017} shows. Thus, a homotopy meaningful variant of it, the \emph{homotopy time-slice axiom} (Definition \ref{defn_WeakTimeSliceAxiom}), must be introduced. We propose two different model structures dealing with this new axiom, as it is summarized in the following theorem (see Subsection \ref{subsect_WeakTimeSliceAxiom}).
\begin{taggedtheorem}{\textbf{A}}\label{thm:A} Let $\Crect^{\perp}$ be an orthogonal category and $\Srect$ a set of maps in $\Crect$.
 Then,
 \begin{itemize}
 	\item the category of AQFTs over $\Crect^{\perp}$ satisfying the homotopy $\Srect$-time-slice axiom admits a Quillen model structure denoted by  $\Qft_{\h\Srect}(\Crect^{\perp})$ with weak equivalences and fibrations defined objectwise (Proposition \ref{prop_STRUCTUREDweakTimeSliceAQFTModel});
 	
 	\item the category $\QFT(\Crect^{\perp})$ admits a Quillen model structure denoted by $\mathsf{L}_{\mathcal{S}}\Qft(\Crect^{\perp})$ presenting the homotopy theory of AQFTs over $\Crect^{\perp}$ satisfying the homotopy $\Srect$-time-slice axiom (Proposition \ref{prop_WeakTimeSliceModel}).
 \end{itemize}
Moreover, both model structures are Quillen equivalent (Theorem \ref{thm_ComparisonWeakTimeSliceModels}). 
\end{taggedtheorem}
Each model category in Theorem \ref{thm:A} deals with the homotopy time-slice axiom emphasizing one aspect of it: as additional structure or as a property of AQFTs. This dual perspective is applied in the last part of Section \ref{sect_TimeLocalization} (from Definition \ref{defn_HomotopyConstancy} to the end of the cited section) to answer \cite[Open Problem 4.14]{benini_higher_2019}. Originally, Benini-Schenkel's problem was of technical nature (asking if a property that they defined on an AQFT captures our homotopy time-slice axiom). We expand on our theory to connect their problem with the \textit{strictification of the homotopy time-slice axiom}. That is, asking if the homotopy time-slice axiom is equivalent to the classical time-slice axiom. By providing Examples \ref{ex_HomotopySconstantTQFT} and \ref{ex_HomotopyConstancy}, we show that this is not the case in general and we identify an obstruction for a positive answer for this strictification problem.

%

As commented before, an important question about AQFTs is if its local behavior, let us say on causal diamonds, completely determines its global one. In this respect, a homotopical local-to-global principle (Definition \ref{defn_FredenhagenLocalToGlobalPrinciple}) that revisits that of Fredenhagen was proposed by Benini, Schenkel and collaborators \cite{benini_higher_2019}. We incorporate such a  principle into a new model structure whose construction is an instance of a general procedure explained in Section \ref{sect_ExtensionModelStructure} (which culminates in Theorem \ref{thm_ExistenceOfExtensionModelAndProperties}). 

In the remainder of this introduction we consider an additional orthogonal category $\Crect_{\diamond}^{\perp}$ which controls the local data in the local-to-global principle. If $\Crect$ consists on suitable open regions in a Lorentzian manifold as before, $\Crect_{\diamond}$ may be seen as the full subcategory of $\Crect$ spanned by causal diamonds.

\begin{taggedtheorem}{\textbf{B}}[Theorem \ref{thm_LocalToGlobalModel}]
Let $\Crect_{\diamond}^{\perp}\hookrightarrow\Crect^{\perp}$ be an inclusion of a full orthogonal subcategory. Then, the category $\QFT(\Crect^{\perp})$ admits a Quillen model structure denoted by $\Qft^{\Crect_{\diamond}}(\Crect^{\perp})$ presenting the homotopy theory of algebraic quantum field theories over $\Crect^{\perp}$ which satisfy the $\Crect_{\diamond}$-local-to-global principle. Moreover, $\Qft^{\Crect_{\diamond}}(\Crect^{\perp})$ is Quillen equivalent to the projective model structure on $\QFT(\Crect^{\perp}_{\diamond})$, denoted $\Qft(\Crect_{\diamond}^{\perp})$. 
\end{taggedtheorem}


Finally, for the combination of both axioms, homotopy time-slice plus homotopical local-to-global principle, we find and compare two new model structures.

\begin{taggedtheorem}{\textbf{C}}\label{thm:C} Let $\Crect_{\diamond}^{\perp}\hookrightarrow\Crect^{\perp}$ be an inclusion of a full orthogonal subcategory and $\Srect$ a set of maps in $\Crect$. Then,
	\begin{itemize}
		\item the category of AQFTs over $\Crect^{\perp}$ satisfying the homotopy $\Srect$-time-slice axiom admits a Quillen model structure denoted by  $\Qft_{\h\Srect}^{\Crect_{\diamond}}(\Crect^{\perp})$ presenting the homotopy theory of those of them adhering to the $\Crect_{\diamond}$-local-to-global principle (Theorem \ref{thm_MixingModelWithStructure});
		
		\item the category $\QFT(\Crect^{\perp})$ admits a Quillen model structure,   $\mathsf{L}_{\mathcal{S}}\Qft^{\Crect_{\diamond}}(\Crect^{\perp})$, presenting the homotopy theory of AQFTs over $\Crect^{\perp}$ satisfying the homotopy $\Srect$-time-slice axiom and the $\Crect_{\diamond}$-local-to-global principle (Theorem  \ref{thm_MixingModelWithProperties}).
	\end{itemize}
	Moreover, both model structures are Quillen equivalent (Theorem \ref{thm_EquivalenceOfMixingModels}). 
	
If the set of maps $\Srect$ lies completely in $\Crect_{\diamond}$, situation that we denote by $\Srect=\Srect_{\diamond}$, then both model structures are also Quillen equivalent to $\Qft_{\h\Srect_{\diamond}}(\Crect_{\diamond}^{\perp})$, which is the model structure established in Theorem \ref{thm:A} for the pair  $(\Crect_{\diamond}^{\perp},\Srect_{\diamond})$ (Theorem \ref{thm_ComparingMixingWithDIAMONDS}).	
\end{taggedtheorem}

Whereas the last part of Theorem \ref{thm:C} could be applied in general, there are physically motivated examples not covered by this statement. For this reason, we show that the analogous result holds for AQFTs over a spacetime with or without timelike boundary (Example \ref{ex_COpenMixingDiamonds}) and for locally covariant conformal AQFTs in two spacetime dimensions  (Example \ref{ex_CLocMixingDiamonds}).

Due to the number of constructions appearing in this paper, we have included a visual glossary of AQFT related model structures and their comparisons in Glossary  \ref{sect_ComparisonFieldTheories}.

The construction of new model structures to present the homotopy theory of AQFTs have interesting consequences other than allowing derived constructions (e.g.\ derived linear quantization  \cite{bruinsma_algebraic_2019}), which, in our view, is one of the main goals in Benini-Schenkel's program. For instance, they can be used to compare different axioms for AQFTs, as we do in this work, to canonically modify a field theory to impose desiderable properties (see Example \ref{ex_ChiralFTsToOrientedTFTs}) or even to establish equivalences between different axiomatizations of quantum field theories, e.g.\ to try to upgrade the main result of \cite{benini_model-independent_2020} to cover gauge theories. In fact, in \cite{benini_strictification_2022}, we make use of the model structures defined in the present work to address the strictification problem for the homotopy time-slice axiom. In contrast with our answer to \cite[Open Problem 4.14]{benini_higher_2019} alluded before, we show that this strictification problem can be solved in a variety of physically interesting examples, e.g.\ for Haag-Kastler-type AQFTs on a fixed globally hyperbolic Lorentzian manifold (with or without time-like boundary), locally covariant AQFTs in one spacetime dimension or locally covariant conformal AQFTs in two spacetime dimensions (see Remark \ref{rem_StrictificationPAPER}). 

Coming back to the general construction of Section \ref{sect_ExtensionModelStructure}, we should mention that it has its own importance in abstract homotopy theory and so several applications are expected. Indeed, in \cite{carmona_factorization_2021}, the homotopy theory of factorization algebras is investigated by these means. 

It is worth noting that most of the results in this work are valid or adaptable for a closed symmetric monoidal model category $\V$ satisfying some general requirements consisting of a combination of hypothesis appearing in Section \ref{sect_ExtensionModelStructure} and  \cite[Section 6]{carmona_factorization_2021}. However, we have chosen to state them for chain-complex valued algebraic field theories over a field $\mathbb{k}$ of characteristic 0 for simplicity and due to connections with the existing literature (see Remark \ref{rem:WhyComplexesChar0}). This restriction is not always made explicit to enforce the cited generality. For instance, Sections \ref{sect_Prelim} and  \ref{sect_ExtensionModelStructure} are written for a general $\V$, whereas the rest of the sections should be interpreted under the assumption $\V=\Ch_{\mathbb{k}}$.

\paragraph{Outline:}
We include a summary of the contents of this work. Section \ref{sect_Prelim} presents basic definitions in the theory of AQFTs and a straightforward construction of the canonical operad developed in \cite{benini_operads_2021} and \cite{bruinsma_coloring_2019}. The rest of the paper is organized depending on which further axiom of AQFTs is under study. Section \ref{sect_TimeLocalization} is devoted to the analysis and discussion of the time-slice axiom and its associated homotopy theories, going from the strict version to the homotopical one, which is seen as part of structure or as a property of AQFTs. The local-to-global principle introduced in the work of Benini, Schenkel and collaborators, and its interaction with the time-slice axiom, is the focus on Section \ref{sect_CausalityColocalization}. To fully understand the constructions in this last section, the content of Section  \ref{sect_ExtensionModelStructure} is necessary. It yields a machine to construct cellularizations of model categories of operadic algebras. Despite this fact, this homotopical discussion can be used as a necessary black box, since it reproduces the local-to-global principle model categorically. At the end, we include a table and a diagram (Glossary  \ref{sect_ComparisonFieldTheories}) condensing the field theory related material obtained in this paper.

\end{section}

\begin{section}{Algebraic quantum field theories}\label{sect_Prelim}
	
\begin{subsection}{Definition and examples}
	
Algebraic quantum field theory is an axiomatic approach to quantum physics which focuses on the assignation of algebras of observables to spacetime regions. The fundamental idea of such objects is that the algebraic structure on the observables is supposed to capture quantum information, so it is not commutative in general, although observables coming from ``spacelike separated" spacetime regions do commute. The notion of orthogonal category formalizes the concept of spacetime regions and the relation of being ``spacelike separated".

\begin{defn}{\cite[Definition 8.2.1]{yau_homotopical_2019}}\label{defn_OrthogonalCategory} An \emph{orthogonal category} $\Crect^{\perp}$ is a pair $(\Crect,\perp)$ given by a small category $\Crect$ with choices of sets of morphisms
	$$
	\perp\;=\Big(\perp(\{\Urect,\Wrect\};\Vrect)\subseteq\Crect(\Urect,\Vrect)\times\Crect(\Wrect,\Vrect)\Big)_{\Urect,\Wrect,\Vrect\in\ob\Crect}
	$$
which satisfy:
\begin{itemize}
\item (Symmetry) $\upf\perp\upg$ if and only if $\upg\perp \upf$,	
\item (Stability under composition) $\uph\upf\upk\perp \uph\upg\upk'$ for $\upf\perp \upg$ and composable maps $\uph,\upk,\upk'$,
\end{itemize}
where we make use of the notation  $\upf\perp\upg$, if $\upf\colon \Urect\to\Vrect\leftarrow\Wrect: \upg$ and $(\upf,\upg)\in\perp(\{\Urect,\Wrect\};\Vrect)$. 

It is said that two $\Crect$-morphisms $\upf$ and $\upg$ are \emph{orthogonal} if $\upf\perp\upg$.
\end{defn}	
	
\begin{notat}
	We call the objects of an orthogonal category  \emph{spacetime regions} to maintain the physical motivation. 
\end{notat}	
	
In turn, the assignment of algebras of observables to spacetime regions is captured by the following definition. Let us fix a target for field theories $\V$, meaning that observables over a spacetime region will belong to $\V$. For concreteness, one can consider $\V=\Ch_{\mathbb{k}}$ the symmetric monoidal category of chain complexes over $\mathbb{k}$. More abstractly, it suffices to assume that $\V$ is a (bicomplete and closed)\footnote{These properties are not essential for the definition of AQFTs but they will be so for later results.} symmetric monoidal category.

\begin{defn}\label{defn_AQFTs} Let $\Crect^{\perp}$ be an orthogonal category.
  An \emph{algebraic quantum field theory} over $\Crect^{\perp}$ is a $\Crect$-diagram of monoids  $\Aalg\colon\Crect\to\Mon(\V)$ such that for any pair of orthogonal maps $(\upf\colon \Urect \to \Vrect) \perp (\upg\colon \Wrect\to \Vrect)$ in $\Crect$, the following square commutes
		$$
		\begin{tikzcd}[ampersand replacement=\&]
		\& \Aalg(\Urect)\otimes\Aalg(\Wrect)  \ar[rd,"\Aalg(\upf)\otimes\Aalg(\upg)"]\ar[ld,"\Aalg(\upf)\otimes\Aalg(\upg)"'] \\ 
		\Aalg(\Vrect)\otimes\Aalg(\Vrect)\ar[rd,"\upmu_{\Vrect}"] \& \& 
		\Aalg(\Vrect)\otimes\Aalg(\Vrect),\ar[dl,"\upmu^{\textup{op}}_{\Vrect}"']\\
		\& \Aalg(\Vrect) 
		\end{tikzcd}	
		$$
		where $\upmu_{\Vrect}$ (resp.\ $\upmu_{\Vrect}^{\textup{op}}=\upmu_{\Vrect}\circ\,\textup{switch}$) denotes the  multiplication map of $\Aalg(\Vrect)$ (resp.\ opposite multiplication).
		The full subcategory of algebraic quantum field theories within $\Fun(\Crect,\Mon(\V))$ is denoted $\QFT(\Crect^{\perp})$.
\end{defn}

Benini, Schenkel and their collaborators proposed this formalization of algebraic quantum field theories in \cite{benini_operads_2021} in order to cover several proposals in the literature. We collect some examples on the following table depending on the underlying orthogonal category (see also \cite{yau_homotopical_2019} and references therein). \vspace*{4mm}

\begin{table}[H]
	\resizebox{\textwidth}{!}{%
		\begin{tabular}{|C{5.5cm}|C{5cm}|C{5cm}|}
			\hline
			\small{\textit{Underlying category} $\Crect$} &
		 \small{\textit{Orthogonality relation} $\perp$} &
	   \small{\textit{Algebraic quantum field theory}} \\
			\hline 
			\hline
			\small{bounded lattice} & \small{$\wedge$-divisors of 0, \break i.e. $\textup{u}\perp\textup{v}$ $\Leftrightarrow$ $\textup{u}\wedge\textup{v}=0$} & \small{Quantum field theories on bounded lattice}\\
			\hline
			\small{oriented n-manifolds with open embeddings preserving orientation} & \small{disjoint image,\break i.e. $\upf\perp\upg$ $\Leftrightarrow$ $\im\upf\cap\im\upg= \cancel{\textup{o}}$} & \small{Chiral conformal quantum field theories}\\
			\hline
			\small{oriented Riemannian n-manifolds with orientation-preserving isometric open embeddings} & \small{disjoint image} & \small{Euclidean quantum field theories}\\
			\hline 
			\small{oriented, time-oriented, globally hyperbolic Lorentzian manifolds with isometric open embeddings preserving orientation and time-orientation whose image is causally convex} & \small{causally disjoint images,\break i.e. $\upf\perp\upg$ if and only if\break there is no causal curve joining $\im\upf$ and $\im\upg$} & \small{Locally covariant quantum field theories} \\
			\hline
			\small{category of regions for a spacetime with timelike boundary} & \small{causally disjoint regions} & \small{Algebraic quantum field theories on spacetime with timelike boundary}\\
			\hline
		\end{tabular}
	}
\end{table}

\end{subsection}

 \begin{subsection}{Operads and homotopy theory}
 	Given an orthogonal category $\Crect^{\perp}$, there is a functorial construction of an operad $\Op_{\Crect}^{\perp}$ whose algebras are algebraic quantum field theories over $\Crect^{\perp}$. More concretely, there is an equivalence of categories
 	$$
 	\Op_{\Crect}^{\perp}\text{-}\Alg(\V)\simeq \QFT(\Crect^{\perp}).
 	$$
    In this section, we recall the definition of an operad and its algebras, how one can define  $\Op_{\Crect}^{\perp}$ (originally introduced in \cite{benini_operads_2021}) with some additional observations which makes this definition more transparent in our view and how one exploits operadic algebras to study the homotopy theory of AQFTs via model structures. 

 \begin{rem} A particularly comprehensible and complete presentation of operads, their algebras and model structures over such can be found in Fresse's two-volume monograph \cite{fresse_homotopy_2017,fresse_homotopy_2017-1}.
\end{rem}

    \paragraph{Operads real quick.} A \emph{(colored symmetric) operad} is a generalization of a category which admits morphisms with multiple inputs. That is, to give an operad $\Op$ is the same as specifying:
    \begin{itemize}
    	\item a collection of colors (also called objects) $\col\Op$,
    	\item Hom-sets $\Op\brbinom{\underline{\upu}}{\upv}$ for each tuple $(\underline{\upu},\upv)\in{\col\Op}^{\times r}\times \col\Op$ and $r\geq 0$,
    	\item permutation actions $\upsigma^*\colon \Op\brbinom{\upu_1,\dots,\upu_r}{\upv}\to \Op\brbinom{\upu_{\upsigma1},\dots,\upu_{\upsigma r}}{\upv}$ for each  $\upsigma\in \Upsigma_r$,
    	\item identities $\id_{\upu}\in\Op\brbinom{\upu}{\upu}$ and composition maps
    	$$
    	\hspace*{-10mm}\left(\circ_j\colon \Op\brbinom{\upu_1,\dots,\upu_r}{\upv}\times \Op\brbinom{\upo_1,\dots,\upo_m}{\upu_j}\to \Op\brbinom{\upu_1,\dots,\upu_{j-1},\upo_1,\dots,\upo_m,\upu_{j+1},\dots,\upu_r}{\upv} \right)_{1\leq j\leq r}.
    	$$
    \end{itemize}
    and these data are  subject to associativity, unitality and equivariance axioms (\cite[Chapter 4]{yau_homotopical_2019}). 
    
    \begin{rem} As for ordinary categories, it is important to take into account size issues by considering small/large operads, but we will not focus on this point. Also, one can consider $\V$-operads as the obvious generalization of $\V$-enriched categories. 
    \end{rem} 
    
    One of the most important reasons to introduce operads is to describe their algebras. An algebra over an operad is the generalization of a diagram or functor out of a category, or in other words, a representation. For example, one can think of a group $\Grect$ as a category with one object and functors out of this category into, let us say, $\mathsf{Vect}_{\mathbb{k}}$ are the same as $\Grect$-representations over $\mathbb{k}$. With this idea in mind, an \emph{algebra} $\Aalg$ over $\Op$ in $\V$ consists of the following data:
    \begin{itemize}
    	\item a collection $(\Aalg(\upu))_{\upu\in\col\Op}$ of objects in $\V$,
    	\item action maps 
    	$
    	\upmu\colon\Op\brbinom{\upu_1,\dots,\upu_r}{\upv}\otimes \bigotimes_{1\leq j\leq r}\Aalg(\upu_j)\rightarrow \Aalg(\upv)
    	$
    	in $\V$.
    \end{itemize} 
    It is required that this data also satisfies  associativity, unity and equivariance axioms. Of course, all algebras over $\Op$ (or $\Op$-algebras) in $\V$ together with structure preserving morphisms between them conform a category $\Op\text{-}\Alg(\V)$. For examples and more details, we refer to \cite{markl_operads_2002,yau_homotopical_2019}.

    \paragraph{Operad controlling  AQFTs in a nutshell.} Let us construct the operad $\Op_{\Crect}^{\perp}$ controlling AQFTs over an orthogonal category $\Crect^{\perp}$, i.e.\ the operad whose algebras are AQFTs over $\Crect^{\perp}$. 
   
 	First, observe that $\QFT(\Crect^{\perp})$ is a full subcategory of $\Fun(\Crect,\Mon(\V))$, which is itself the category of algebras over a certain operad. To find this first operad, note that the category of monoids in $\V$ is equivalent to the category of $\mathsf{uAss}$-algebras in $\V$, where $\mathsf{uAss}$ denotes the operad of unital associative algebras. Looking at $\Crect$ as an operad as well, we have the identification
 	$$
 	\Fun(\Crect,\Mon(\V))\cong \Crect\text{-}\Alg\big(\mathsf{uAss}\,\text{-}\Alg(\V)\big).
 	$$ 
 	The Boardman-Vogt tensor product (see \cite{weiss_operads_2011} or the original source \cite[Section II.3]{boardman_homotopy_1973}), denoted $\otimes_{\mathtt{BV}}$, combines $\Crect$ and $\mathsf{uAss}$ to obtain another operad, $\Crect\otimes_{\mathtt{BV}}\,\mathsf{uAss}$, which comes with the identification 
 	$$
 	\big(\Crect\underset{\mathtt{BV}}{\otimes}\,\mathsf{uAss}\big)\text{-}\Alg(\V)\cong\Crect\text{-}\Alg\big(\mathsf{uAss}\,\text{-}\Alg(\V)\big). 
 	$$ 
    In other words, the operad governing $\Crect$-diagrams of monoids in $\V$ is just $ \Crect\otimes_{ \mathtt{BV}}\,\mathsf{uAss}$.
 	
 	Within such diagrams, algebraic field theories are those for which certain operations coincide. More precisely, looking at $\Aalg\in\QFT(\Crect^{\perp})$ as an algebra over the Boardman-Vogt tensor product, it must satisfy
 	$$
 	\begin{tikzpicture}[level distance=6mm, sibling distance=8mm]
 	\draw (0,0.9) node {$\Aalg$};
 	\draw (1.2,0) node {\scriptsize $\Vrect$} [grow'=up] 
 	child{node (B) [rectangle, draw=black!50, fill=black!20, inner sep=0,outer sep =0, minimum size = 5mm]  {$\scriptstyle \upmu_{\Vrect}$}  edge from parent {} 
 		child{node (C) [circle, draw=black!50, fill=black!20, inner sep=0,outer sep =0, minimum size = 3mm]  {$\scriptstyle \upf$} 
 			child{node {\scriptsize $\Urect$}}
 		}        child{node (D) [circle, draw=black!50, fill=black!20, inner sep=0,outer sep =0, minimum size = 3mm]  {$\scriptstyle \upg$} 
 			child{node {\scriptsize $\Wrect$}}
 	}};
    \draw (2.4,0.9) node {$=$};
    \draw (3,0.9) node {$\Aalg$};
    \draw (4.2,0) node {\scriptsize $\Vrect$} [grow'=up] 
    child{node (B) [rectangle, draw=black!50, fill=black!20, inner sep=0,outer sep =0, minimum size = 5mm]  {$\scriptstyle \upmu^{\text{op}}_{\Vrect}$}  edge from parent {} 
    	child{node (C) [circle, draw=black!50, fill=black!20, inner sep=0,outer sep =0, minimum size = 3mm]  {$\scriptstyle \upf$} 
    		child{node {\scriptsize $\Urect$}}
    	}        child{node (D) [circle, draw=black!50, fill=black!20, inner sep=0,outer sep =0, minimum size = 3mm]  {$\scriptstyle \upg$} 
    		child{node {\scriptsize $\Wrect$}}
    }};
    \draw [decorate,decoration={brace,amplitude=10pt},xshift=-0.5cm,yshift=0pt]
    (1.2,-0.15) -- (1.2,2.05) node [black,midway,xshift=-0.6cm]
    {};
    \draw [decorate,decoration={brace,amplitude=10pt,mirror},xshift=+0.6cm, yshift=0pt]
    (1.2,-0.15) -- (1.2,2.05) node [black,midway,xshift=+0.6cm] {};
 	\draw [decorate,decoration={brace,amplitude=10pt},xshift=-0.5cm,yshift=0pt]
 	(4.2,-0.15) -- (4.2,2.05) node [black,midway,xshift=-0.6cm]
 	{};
 	\draw [decorate,decoration={brace,amplitude=10pt,mirror},xshift=+0.6cm, yshift=0pt]
 	(4.2,-0.15) -- (4.2,2.05) node [black,midway,xshift=+0.6cm] {};
 	\end{tikzpicture}
 	$$
 	for all $\upf\perp\upg$ in $\Crect$\footnote{The use of trees to described operads is quite useful and it is justified in, for example,  \cite{markl_operads_2002} or \cite{yau_homotopical_2019}.}. Therefore, one can present this algebraic structure by constructing a suitable quotient of $\Crect\otimes_{ \mathtt{BV}}\,\mathsf{uAss}$ which identifies those operations. This can be done by introducing a free operad $\mathfrak{F}(\mathsf{R}_{\perp})$ that selects these operations, and coequalizing the corresponding maps
 	$$
 \begin{tikzcd}[ampersand replacement=\&]
\mathfrak{F}(\mathsf{R}_{\perp})\ar[r, shift left=.75ex, "\upmu_{\perp}"]\ar[r, shift right=.75ex, "\upmu^{\text{op}}_{\perp}"'] \& \Crect\underset{\mathtt{BV}}{\otimes}\,\mathsf{uAss}\ar[r, "\text{coeq}"] \& \Op_{\Crect}^{\perp}.
 \end{tikzcd}
 	$$ 
 	The underlying symmetric sequence of the free operad on the left is simply
 	$$
 	\mathsf{R}_{\perp}(\{\Urect_1\dots \Urect_{\upn}\};\Vrect)=
 	\left\lbrace\begin{matrix}
 	\coprod_{\perp(\{\Urect_1,\Urect_{2}\};\Vrect)} \mathbb{I} & \text{ if }\;\;\upn=2\\\\
 	\cancel{\textup{o}} & \text{ otherwise}
 	\end{matrix} \right.,
 	$$
 	and by freeness, the map $\upmu_{\perp}$ corresponds to choosing the set of operations 
 		$$
 	\begin{tikzpicture}[level distance=6mm, sibling distance=8mm]
 	\draw (0,0) node {\scriptsize $\Vrect$} [grow'=up] 
 	child{node (B) [rectangle, draw=black!50, fill=black!20, inner sep=0,outer sep =0, minimum size = 5mm]  {$\scriptstyle \upmu_{\Vrect}$}  edge from parent {} 
 		child{node (C) [circle, draw=black!50, fill=black!20, inner sep=0,outer sep =0, minimum size = 3mm]  {$\scriptstyle \upf$} 
 			child{node {\scriptsize $\Urect$}}
 		}        child{node (D) [circle, draw=black!50, fill=black!20, inner sep=0,outer sep =0, minimum size = 3mm]  {$\scriptstyle \upg$} 
 			child{node {\scriptsize $\Wrect$}}
 	}};
   \draw (4.4,1.1) node {where $\upmu_{\Vrect}$ is associated to $\upmu_2\in\mathsf{uAss}(2)$};
   \draw (4.4,0.6) node {and $\upf\perp\upg$ are orthogonal in $\Crect$};
    \draw [decorate,decoration={brace,amplitude=10pt},xshift=-0.5cm,yshift=0pt]
   (-0.5,-0.15) -- (-0.5,2.05) node [black,midway,xshift=-0.6cm]
   {};
   \draw [decorate,decoration={brace,amplitude=10pt,mirror},xshift=+0.6cm, yshift=0pt]
   (7.5,-0.15) -- (7.5,2.05) node [black,midway,xshift=+0.6cm] {};
 	\end{tikzpicture}
 	$$
 	in $\Crect\otimes_{ \mathtt{BV}}\,\mathsf{uAss}$ (and analogously for $\upmu^{\text{op}}_{\perp}$, just replace $\upmu_{\Vrect}$ by $\upmu_{\Vrect}^{\text{op}}$ above).
 	
    The recognition of the Boardman-Vogt tensor product in the construction of $\Op_{\Crect}^{\perp}$ is not vacuous. It leads to a broad generalization of the functorial construction
 		$$
 		\mathsf{OrthCat}\times\mathsf{Operad}_{\{*\}}^{2\pt}(\V)\longrightarrow \mathsf{Operad}(\V), \;\;\; (\Crect^{\perp},\OpP)\mapsto \OpP^{(\upr_1,\upr_2)}_{\Crect^{\perp}}
 		$$
 	discussed in \cite{bruinsma_coloring_2019, bruinsma_algebraic_2019}, where we denote by $\mathsf{OrthCat}$ the category of orthogonal categories and by $\mathsf{Operad}(\V)$ that of $\V$-operads. Now, it can be seen as the composition of:
 	\begin{itemize}
 		\item the Boardman-Vogt tensor $\otimes_{\mathtt{BV}}\colon\mathsf{Cat}\times \mathsf{Operad}(\V)\to \mathsf{Operad}(\V)$, to construct the operad $\OpP_{\Crect}\cong \Crect\otimes_{\mathtt{BV}}\OpP$ encoding $\Crect$-diagrams of $\OpP$-algebras (see \cite{weiss_operads_2011});
 		\item a colimit construction to identify operations in $\OpP_{\Crect}\cong \Crect\otimes_{\mathtt{BV}}\OpP$ as ruled out by the orthogonality relation in $\Crect$.
 	\end{itemize}
    Hence, it is admissible to consider orthogonal categories enriched in $\V$ and more general quotients than those coming from a monochromatic bipointed operad (following the notation on \cite{bruinsma_algebraic_2019}). For example, the category $\mathsf{Loc}$ of globally hyperbolic Lorentzian manifolds \cite{fewster_algebraic_2015} could be endowed with some topological structure and the bipointed operad could have multiple colors. We expect to exploit this flexibility in subsequent work.
    
  \begin{rem}
The recognition of the Boardman-Vogt tensor product provides a canonical presentation by generators and relations of the operad (see \cite{weiss_operads_2011}). As expected, it coincides with the presentation by generators and relations  of $\Op_{\Crect}^{\perp}$ given in \cite{benini_operads_2021}.
  \end{rem}

 	 From now on, operads and algebras are considered in $\V=\Ch_{\mathbb{k}}$, with $\mathbb{k}$ a field of characteristic $0$, unless otherwise specified. This choice also affects our notation, e.g. $\Set$-operads are consider in $\Ch_{\mathbb{k}}$ via the strong monoidal functor $\Set\to\Ch_{\mathbb{k}}$, $\X\mapsto\bigoplus_{\X}\mathbb{k}$.

 	\paragraph{A brief recap on model structures and homotopy theory.} The recognition of AQFTs as operadic algebras was a fundamental step in the program of Benini-Schenkel and collaborators which allowed them to perform physically motivated constructions in a derived way, such as the linear quantization of \cite{bruinsma_algebraic_2019} (see also  \cite{benini_higher_2019} for more examples), that is, yielding the same result for \emph{(weakly) equivalent} AQFTs. Recall that two AQFTs over $\Crect^{\perp}$, $\Aalg$ and $\Balg$, are weakly equivalent if there exists a zigzag of morphisms of AQFTs
 	$$
 	\Aalg\leftarrow\bullet\rightarrow\bullet\leftarrow\cdots\rightarrow\bullet\leftarrow \Balg
 	$$ 
 	which are quasi-isomorphisms when evaluated on every spacetime region $\Urect\in\Crect$.
    This notion is more flexible than being isomorphic and it is necessary when homotopical/homological phenomena is considered, such as field theories with gauge symmetries, to identify theories with the same physical content. 
 	
 	An especially powerful mathematical device designed to deal with this situation: a category, e.g. $\QFT(\Crect^{\perp})$, in which we are interested on notions up to (a specified version of) \emph{weak equivalence} is that of \emph{model category}. The concrete definition of a model category is lengthy and not very inspirational in a first sight. For this reason, we will just briefly revisit some ideas about this notion and refer interested readers to \cite{barwick_left_2010,hirschhorn_model_2003,hovey_model_1999} for more details. 
 	
 	A \emph{model category} or \emph{model structure} $\M$  consists of a tuple $(\Mrect,\Wrect,\Cof,\Fib)$ where $\Mrect$ is a (bicomplete) category and $\Wrect,\Cof,\Fib$ are classes of morphisms in $\Mrect$ subject to several axioms. Arguably, the most important part of this data is $(\Mrect,\Wrect)$ because it dictates what objects are we interested in (given by $\Mrect$) and the specific notion of \emph{weak equivalence} between them that we are fixing (a map in $\Wrect$ is said to be a weak equivalence and two objects are weak equivalent if they can be connected by a zigzag of maps in $\Wrect$). We can refer to this as the \emph{homotopy theory associated to} $(\Mrect,\Wrect)$. In abstract terms, one is concerned about the category $\Ho\M$ obtained from $\Mrect$ by formally adding inverses for the maps in $\Wrect$.  
 	
 	The remaining classes of morphisms $\Cof$ and $\Fib$ yield additional data that allow us to perform constructions and to control notions up to weak equivalence. More concretely, the axioms of a model category mimic the classical lifting and factorization relations between (acyclic) cofibrations and (acyclic) fibrations encountered within the homotopy theory of topological spaces: very concisely, $(\Cof\cap \Wrect,\Fib)$ and $(\Cof,\Fib\cap\Wrect)$ are weak factorization systems in $\Mrect$; and those are enough to explore the \emph{homotopy theory associated to} $(\Mrect,\Wrect)$. 
 	
 	\begin{notat} We will say that the model category $\M$ \emph{presents the homotopy theory of} $(\Mrect,\Wrect)$.
 	\end{notat}
 	 
 	 We should also mention that model categories can be connected via \emph{Quillen adjunctions} (also called \emph{Quillen pairs}).
 	 
 	 \begin{defn}
 	  A Quillen adjunction $\Frect\colon \M\rightleftarrows \M':\Grect$ is a pair of adjoint functors $\Frect\dashv\Grect$ between the underlying categories associated to $\M$ and $\M'$ satisfying one of the following equivalent conditions:
 	 \begin{itemize}
 	 	\item $\Frect$ preserves cofibrations ($\Cof$) and acyclic cofibrations  ($\Cof\cap\Wrect$);
 	 	\item $\Grect$ preserves fibrations ($\Fib$) and acyclic fibrations ($\Fib\cap\Wrect$);
 	 	\item $\Frect$ preserves cofibrations and $\Grect$ preserves fibrations;
 	 	\item $\Frect$ preserves acyclic cofibrations and $\Grect$ preserves acyclic fibrations.
 	 \end{itemize}
 	 \end{defn}
 
 	 The main consequence of $\Frect\dashv\Grect$ being a Quillen adjunction is that there is an associated \emph{derived adjunction} $\mathbb{L}\Frect\dashv\mathbb{R}\Grect$. That is, there are canonical modifications $\mathbb{L}\Frect$ of $\Frect$ and $\mathbb{R}\Grect$ of $\Grect$ that respect  being weak equivalent (derived constructions as mentioned at the beginning of this discussion) and which formally conform an adjunction 
 	 $
 	 \mathbb{L}\Frect\colon \Ho\M\rightleftarrows\Ho\M':\mathbb{R}\Grect.
 	 $
 	 A \emph{Quillen equivalence} is a Quillen adjunction  whose associated derived adjunction is an adjoint equivalence $\Ho\M\simeq \Ho\M'$. Hence, when two model categories are connected by a (zigzag of) Quillen equivalence(s), they \emph{present the same homotopy theory}.
 	 
 	 A particularly useful tool in the world of model categories is that of \emph{left} or \emph{right Bousfield localization}. These processes consist of a modification of the classes of maps $(\Wrect,\Fib)$ (for left Bousfield localization) or $(\Wrect,\Cof)$ (for right Bousfield localization) of a given model category to present different homotopy theories. In this work we will make extensive use of them and we recommend the interested reader to consult \cite{barwick_left_2010}, \cite{beke_sheafifiable_2000} or \cite{hirschhorn_model_2003} for definitions, notation and motivation.
 	 
 	 \paragraph{Homotopy theory of AQFTs.} Now we want to equip the category $\QFT(\Crect^{\perp})$ with a model structure to study field theories up to weak equivalence, i.e.\ with the same physical content. Since we are considering $\V=\Ch_{\mathbb{k}}$ as the target of field theories, there is a first choice of model structure on $\V$ that will be fixed from now on.
 	 
 	  	\begin{assump} The category $\V=\Ch_{\mathbb{k}}$ of complexes over a field $\mathbb{k}$ of characteristic $0$  is endowed with the projective model structure, i.e.\ the class of weak equivalences is that of quasi-isomorphisms and the class of fibrations is that of epimorphisms (see \cite[Section 2.3]{hovey_model_1999}).
 	  \end{assump}
 	 
 	 The recognition of AQFTs as operadic algebras serves to use this model structure on $\V$ to construct another one on $\QFT(\Crect^{\perp})\cong\Op_{\Crect}^{\perp}\text{-}\Alg(\V)$ by \cite[Theorem 4.3.3]{fresse_homotopy_2017-1}: 
 	 
 	 \begin{thm}{\cite[Theorem 3.10]{benini_homotopy_2019}}\label{thm_ProjectiveModelOnPlainAQFTs} The category $\QFT(\Crect^{\perp})$ admits the projective model structure, denoted  $\Qft(\Crect^{\perp})$, which is completely  characterized by:
 	 \begin{itemize}
 	 \item its weak equivalences (called projective weak equivalences) are those maps in $\QFT(\Crect^{\perp})$ which are quasi-isomorphisms when evaluated over every spacetime region $\Urect\in\Crect$,  
 	 \item its fibrations (called projective fibrations) are those maps in $\QFT(\Crect^{\perp})$ which are epimorphisms when evaluated over every spacetime region $\Urect\in\Crect$.
   	 \end{itemize}
 	 \end{thm}
 	 
 	 \begin{rem}\label{rem:WhyComplexesChar0} We decided to work with complexes over a field $\mathbb{k}$ of characteristic 0 mainly for its connection with the foundational work of Benini, Schenkel and collaborators, and because of the simplicity of the homotopy theory of operads and algebras in this context. For instance, the strictification theorem for operadic algebras in $\Ch_{\mathbb{k}}$ is one of the principal advantages of working within this setting. This theorem asserts that operadic algebras over an operad presents the homotopy theory of their homotopy coherent version, because such operads are automatically $\Upsigma$-cofibrant by \cite[Proposition 2.11]{benini_homotopy_2019} and because weak equivalences of $\Upsigma$-cofibrant operads induce Quillen equivalences between the corresponding model categories of operadic algebras \cite[Theorem 2.10]{benini_homotopy_2019}. 
 
  	 Suitable replacements of the operad $\Op_{\Crect}^{\perp}$ worked out by Yau \cite{yau_homotopical_2019} and Benini-Schenkel-Woike \cite{benini_homotopy_2019} lead to homotopy coherent variants of algebraic quantum field theories, which may have a priori a richer homotopy theory than the strict version explored in this work.
  	 Their contribution consists on choosing a weak equivalent operad $\Op_{\Crect,\infty}^{\perp}\xrightarrow{\sim}\Op_{\Crect}^{\perp}$, which is more cofibrant than the original one in some precise sense, and looking at $\QFT_{\infty}(\Crect^{\perp})=\Op_{\Crect,\infty}^{\perp}\text{-}\Alg(\V)$. However, the strictification theorem states that no more homotopical information is gained if we work within $\V=\Ch_{\mathbb{k}}$, since we have a Quillen equivalence
  	 $$
  	 \Op_{\Crect,\infty}^{\perp}\text{-}\Alg(\V)\simeq \Op_{\Crect}^{\perp}\text{-}\Alg(\V).
  	 $$ 
  	 
  	 Nevertheless, as briefly commented in the introduction, our techniques could be adapted for homotopy coherent algebraic quantum field theories as well, leading to parallel results.  
  	 The importance of this generalization is that it is possible to replicate our homotopical results in more general contexts, for instance allowing complexes over a general base ring $\mathsf{R}$ such as $\mathbb{k}\llbracket\hbar\rrbracket$, some flavour of spectra, etc.  
 	 \end{rem}
 	 
%
%
%
%
%
%
%

 \end{subsection}

\end{section}

\begin{section}{Time-slice localization}\label{sect_TimeLocalization}
	 Given an algebraic quantum field theory, a way to introduce dynamics into the picture is via the so called time-slice axiom on the theory (see \cite{benini_higher_2019} or \cite{fewster_algebraic_2015} for a  motivation of the concept). Formally, one identifies a set of maps $\Srect$ in $\Crect$ that control ``evolution" between spacetime-regions, e.g.\ Cauchy morphisms. Field theories with a nice behaviour with respect to this ``evolution" are those $\Aalg\in\QFT(\Crect^{\perp})$ satisfying that $\Aalg(\upf)$ is an isomorphism for any $\upf\in\Srect$.  Such field theories are said to satisfy the $\Srect$\emph{-time-slice axiom}. In this section, we propose two methods to study the homotopy theory of algebraic quantum field theories satisfying the time-slice axiom (or its relaxed version, see Definition \ref{defn_WeakTimeSliceAxiom}). 
 
\begin{defn}\label{defn_QFTsSatisfyingTimeSliceAxiom} A field theory
$\Aalg\in \QFT(\Crect^{\perp})$ is said to satisfy the $\Srect$-\emph{time-slice axiom} or to be $\Srect$-\emph{constant} if $\Aalg(\upf)$ is an isomorphism for any $\upf\in\Srect$. The full subcategory of $\QFT(\Crect^{\perp})$ spanned by $\Srect$-constant field theories is denoted $\QFT_{\Srect}(\Crect^{\perp})$.
\end{defn}


\begin{subsection}{Strict time-slice axiom}\label{subsect_StrictTimeSliceAxiom}	
	
	Let us begin with a brief reminder of the time-slice axiom discussion in \cite{benini_higher_2019, benini_homotopy_2019, benini_operads_2021}. There, it is noticed that looking at an algebraic field theory as a diagram $\Aalg\colon\Crect\to\Mon(\V)$, $\Srect$-constancy can be expressed by saying that the diagram factors uniquely through the localization of $\Crect$ at $\Srect$, denoted $\Crect\to \Crect_{\Srect}$\footnote{The localization of $\Crect$ at $\Srect$ is the ``initial" functor $\Crect\to\Crect_{\Srect}$ among those $\Crect\to \Drect$ that send $\Srect$ to isomorphisms.}. Indeed, the localization $\Crect_{\Srect}$ can be endowed with a canonical orthogonality relation, the so-called pushforward orthogonality relation $\perp_{\Srect}$; it is constructed in \cite[Lemma 3.19]{benini_operads_2021}. The importance of $\Crect_{\Srect}^{\perp_{\Srect}}$ is that it encodes the algebraic structure of $\Srect$-constant field theories. More precisely, in \cite[Lemma 3.20, Proposition 4.4]{benini_operads_2021}, it is showed that the induced morphism $\ell\colon\Op_{\Crect}^{\perp}\to \Op_{\Crect_{\Srect}}^{\perp_{\Srect}}$ yields an equivalence of categories of algebraic field theories
	$$
	\ell^*\colon\QFT(\Crect_{\Srect}^{\perp_{\Srect}})\overset{\sim}{\longrightarrow} \QFT_{\Srect}(\Crect^{\perp}).
	$$  
    Thus, a change of orthogonal category permits the recognition of $\Srect$-constancy as structure and not as a property on field theories. Consequently, the homotopy theory of $\Srect$-constant field theories may be  presented by the projective model structure over $ \Op_{\Crect_{\Srect}}^{\perp_{\Srect}}\text{-}\Alg\simeq \QFT(\Crect_{\Srect}^{\perp_{\Srect}})$  (see \cite[Theorems 3.8 and 3.10]{benini_homotopy_2019}).
    \begin{notat}\label{notat_ProjModelonAQFTsWithStrictTimeSliceAxiom}
    	$\Qft_{\Srect}(\Crect^{\perp})$ denotes the projective model structure on $\Op_{\Crect_{\Srect}}^{\perp_{\Srect}}\text{-}\Alg$ of Theorem \ref{thm_ProjectiveModelOnPlainAQFTs}.
    \end{notat}
     Observe that this homotopy theory only deals with field theories satisfying the $\Srect$-time-slice axiom strictly, as introduced above, and this fact excludes algebraic field theories coming from gauge theory (see \cite[Section 3.4]{benini_higher_2019}). The way to include such examples will be a relaxation of the strict $\Srect$-time-slice axiom. In order to do so, we will need a more canonical description of $\Op_{\Crect_{\Srect}}^{\perp_{\Srect}}$, motivated by the algebraic field theories that it classifies.
    	
   \begin{prop}\label{prop_LocalizingSpacetimeRegionsOrOperadCompletely}
   	Let $\Crect^{\perp}$ be an orthogonal category and $\Crect_{\Srect}^{\perp_{\Srect}}$ its localization at $\Srect$. Then, the canonical morphism $\Op_{\Crect}^{\perp}\to \Op_{\Crect_{\Srect}}^{\perp_{\Srect}}$ induces an equivalence of operads
   	$$
   	\big(\Op_{\Crect}^{\perp}\big)_{\Srect}\overset{\sim}{\longrightarrow}\Op_{\Crect_{\Srect}}^{\perp_{\Srect}},
   	$$
   	where the left-hand side denotes the $\Set$-localization of the operad  $\Op_{\Crect}^{\perp}$ at $\Srect$.
   \end{prop} 	
    	\begin{proof}
    	The claim follows from the following chain of isomorphisms:
    	\begin{align*}
    	\mathsf{Operad}\big(\Op_{\Crect_{\Srect}}^{\perp_{\Srect}},\OpP\big) & \overset{(i)}{\cong} 
        \mathsf{equalizer}\Big[\mathsf{Operad}(\Crect_{\Srect}\underset{\mathtt{BV}}{\otimes}\,\mathsf{uAss},\OpP)\rightrightarrows\mathsf{Operad}(\mathfrak{F}(\mathsf{R}_{\perp_{\Srect}}),\OpP)\Big]\\
         & \overset{(ii)}{\cong} 
        \mathsf{equalizer}\Big[\mathsf{Operad}(\Crect_{\Srect},\OpP^{\mathsf{uAss}})\rightrightarrows\mathsf{Operad}(\mathfrak{F}(\mathsf{R}_{\perp_{\Srect}}),\OpP)\Big]\\
        & \overset{(iii)}{\cong} 
        \mathsf{equalizer}\Big[\mathsf{Operad}_{\Srect}(\Crect,\OpP^{\mathsf{uAss}})\rightrightarrows\mathsf{Operad}(\mathfrak{F}(\mathsf{R}_{\perp}),\OpP)\Big]\\
        & \overset{(iv)}{\cong} 
        \mathsf{equalizer}\Big[\mathsf{Operad}_{\Srect}(\Crect\underset{\mathtt{BV}}{\otimes}\,\mathsf{uAss},\OpP)\rightrightarrows\mathsf{Operad}(\mathfrak{F}(\mathsf{R}_{\perp}),\OpP)\Big]\\
        & \overset{(v)}{\cong}
        \mathsf{Operad}_{\Srect}(\Op_{\Crect}^{\perp},\OpP).
    	\end{align*}
    	We have used: $(i)$ description of the operad $\Op_{\Crect}^{\perp}$ as a quotient of the Boardman-Vogt tensor product; $(ii)$ Boardman-Vogt tensor product is left-adjoint to internal hom in $\mathsf{Operad}$; $(iii)$ universal property of the localization $\Crect_{\Srect}$ on the first factor of the equalizer. On the second factor, we apply that the orthogonal relation $\perp_{\Srect}$ is generated by $\perp$. That is, by the explicit description for the pushforward orthogonal relation of \cite[Lemmas 3.19]{benini_operads_2021} and the proof of \cite[Lemma 3.20]{benini_operads_2021}, we see that the equalizer requiring $\perp$-commutativity is the same as the one that requires $\perp_{\Srect}$-commutativity; $(iv)$ same adjunction as in $(ii)$; $(v)$ universal property of localization together with $(i)$.
    	\end{proof}
    	
    	\begin{rem} Proposition \ref{prop_LocalizingSpacetimeRegionsOrOperadCompletely} reproves that the morphism $\ell\colon \Op_{\Crect}^{\perp}\to \Op_{\Crect_{\Srect}}^{\perp_{\Srect}}$ induces an equivalence of categories $\ell^*\colon\QFT(\Crect_{\Srect}^{\perp_{\Srect}})\overset{\sim}{\longrightarrow} \QFT_{\Srect}(\Crect^{\perp})$. See the first paragraph of Subsection \ref{subsect_StrictTimeSliceAxiom}.
    	\end{rem}
\end{subsection}

\begin{subsection}{Homotopy time-slice axiom}\label{subsect_WeakTimeSliceAxiom}
	
	Motivated by toy models explained in \cite{benini_higher_2019}, it is natural to relax the $\Srect$-time-slice axiom to a homotopical variant.
	\begin{defn}\label{defn_WeakTimeSliceAxiom} An algebraic quantum field theory $\Aalg\in\QFT(\Crect^{\perp})$ satisfies the \emph{homotopy $\Srect$-time-slice axiom}, or equivalently it is \emph{homotopy} $\Srect$-\emph{constant}, if $\Aalg(\upf)$ is a weak equivalence whenever $\upf\in\Srect$. We denote by $\QFT_{\h\Srect}(\Crect^{\perp})$ the full subcategory of $\QFT(\Crect^{\perp})$ spanned by homotopy $\Srect$-constant field theories. 
	\end{defn}
	
	We propose two approaches to study the homotopy theory of these algebraic quantum field theories and we show that these two approaches are equivalent in a precise sense. 
	
	\paragraph{Structural approach:} It consists on enhancing the operad $\Op_{\Crect}^{\perp}$ to encode the structure that turns the maps in $\Srect$ to weak equivalences. Due to Proposition \ref{prop_LocalizingSpacetimeRegionsOrOperadCompletely}, the idea is similar to the one used in \cite{benini_operads_2021} to study $\Srect$-constant field theories. The difference is that we use a homotopical variant of the localization:

	\begin{notat}
		We will denote by  $\mathfrak{L}_{\Srect}\Op_{\Crect}^{\perp}$ the \emph{homotopical localization} of $\Op_{\Crect}^{\perp}$ at $\Srect$.
	\end{notat}
	
    Its fundamental property and why we are interested in such a device is given by:
	 \begin{prop}\label{prop_OperadGoverningWeaklyConstantFieldTheories}
		There is a canonical equivalence
		$
	\Ho\QFT_{\h\Srect}(\Crect^{\perp})\simeq \Ho\big(\mathfrak{L}_{\Srect}\Op_{\Crect}^{\perp}\text{-}\Alg(\V)\big).
	$
	\end{prop}

	\begin{rems}$\,$
		\begin{enumerate}
			\item[(1)] The existence of  $\mathfrak{L}_{\Srect}\Op_{\Crect}^{\perp}$ can be derived from the theory of either simplicial operads or $\infty$-operads in the sense of Lurie, since $\Op_{\Crect}^{\perp}$ is an operad in sets. For instance, it can be  constructed as a homotopy pushout in the category of simplicial operads or dg-operads as how it is constructed the homotopy localization of dg-categories in   \cite{toen_homotopy_2007}.
			\item[(2)]  The operad $\mathfrak{L}_{\Srect}\Op_{\Crect}^{\perp}$ satisfies a homotopical universal property, so it is unique up to equivalence: $\Op_ {\Crect}^{\perp}\to \mathfrak{L}_{\Srect}\Op_{\Crect}^{\perp}$ is initial (in the $\infty$-sense) among maps of (simplicial) operads $\Op_{\Crect}^{\perp}\to \OpP$ that send maps in $\Srect$ to equivalences.
			 Since we do not need any particular presentation, we prefer to state our results making no reference to specific models for  $\mathfrak{L}_{\Srect}\Op_{\Crect}^{\perp}$.
			\item[(3)]  The localization $\mathfrak{L}_{\Srect}\Op_{\Crect}^{\perp}$ is performed directly at the level of operads, in contrast to the strict case, where the localization is done for the orthogonal category $\Crect_{\Srect}^{\perp_{\Srect}}$, see Subsection \ref{subsect_StrictTimeSliceAxiom}. With this observation, we want to stress that the orthogonality relation on $\Crect$ does not need to be pushed forward.
		\end{enumerate}
	\end{rems}


  
	
    Since homotopy $\Srect$-constant field theories are presented as algebras over the operad  $\mathfrak{L}_{\Srect}\Op_{\Crect}^{\perp}$, the projective model structure on these algebras, which we denote appealingly by $\Qft_{\h\Srect}(\Crect^{\perp})$, presents the homotopy theory of those field theories. This model structure can be also constructed using \cite[Theorem 4.3.3]{fresse_homotopy_2017-1} as explained for Theorem \ref{thm_ProjectiveModelOnPlainAQFTs}, but it demands slightly more than that theorem because $\mathfrak{L}_{\Srect}\Op_{\Crect}^{\perp}$ is not an operad in sets (in particular, it is not of type $\Op_{\Drect}^{\perp}$ for an orthogonal category $\Drect^{\perp}$). 
	
	\begin{prop}\label{prop_STRUCTUREDweakTimeSliceAQFTModel}
	The homotopy theory of $\QFT_{\h\Srect}(\Crect^{\perp})$ is presented by the projective model structure $\Qft_{\h\Srect}(\Crect^{\perp})$ with weak equivalences and fibrations defined pointwise. 
	\end{prop}
	
	\begin{rem}
	Recall that a map of algebraic quantum field theories $\Aalg\to \Balg$ is said to be pointwise in some class of maps if $\Aalg(\Urect)\to\Balg(\Urect)$ belongs to this class for any space-time region $\Urect\in\Crect$.
	\end{rem}
	
	Let us now explain the second alternative.
	
	\paragraph{Property-based approach:} Field theories in $\QFT_{\h\Srect}(\Crect^{\perp})$ have underlying algebraic structure encoded by $\mathfrak{L}_{\Srect}\Op_{\Crect}^{\perp}$. Now we want to relax the algebraic structure by providing a model structure on $\QFT(\Crect^{\perp})$ which also presents the homotopy theory of homotopy $\Srect$-constant field theories. For this purpose, we use left Bousfield localization \cite{barwick_left_2010} for a suitable set of maps.
	
	A suitable choice of a localizing set of maps in $\QFT(\Crect^{\perp})$ is given in \cite[Section 6]{carmona_factorization_2021}. Using the same methods, it is possible to construct a set of maps $\mathcal{S}$ in $\QFT(\Crect^{\perp})$ \emph{representing} $\Srect\subseteq\Crect$ in such a way that the following result holds by the same arguments. Very concisely, one takes
	$$
	\mathcal{S}=\begin{Bmatrix}
		\mathfrak{t}_{\sharp}(\upf^*[r])\colon \mathfrak{t}_{\sharp}\big(\Crect(\Urect,\star)\otimes \mathbb{k}[r] \big)\to \mathfrak{t}_{\sharp}\big(\Crect(\Vrect,\star)\otimes \mathbb{k}[r] \big)\quad \text{for }\upf\colon \Vrect\to\Urect \text{ in }\Srect
	\end{Bmatrix},
	$$
	where $\mathfrak{t}_{\sharp}\colon \Fun(\Crect,\Ch_{\mathbb{k}})\rightleftarrows \QFT(\Crect^{\perp}):\mathfrak{t}^*$ is the adjoint pair associated to the obvious map  $\mathfrak{t}\colon\Crect\to\Op_{\Crect}^{\perp}$ and $\Crect(\Urect,\star)\otimes\mathbb{k}[r]$ denotes the functor $\Crect\to \Ch_{\mathbb{k}}$, $\Wrect\mapsto \bigoplus_{\Crect(\Urect,\Wrect)}\mathbb{k}[r]$.
	
	\begin{prop}\label{prop_CharacterizationOfTIMESLICELOCALObjects}
	A field theory $\Aalg\in\QFT(\Crect^{\perp})$ is $ \mathcal{S}$-local if and only if it satisfies the homotopy $\Srect$-time-slice axiom.
	\end{prop}
	\begin{proof} See \cite[Theorem 6.5]{carmona_factorization_2021}.
	\end{proof}

	 We are now ready to give the second model structure to present the homotopy theory of homotopy $\Srect$-constant field theories.  
	
	\begin{prop}\label{prop_WeakTimeSliceModel} 
		The model structure $\Qft(\Crect^{\perp})$ admits the left Bousfield localization at $\mathcal{S}$, denoted $\mathsf{L}_{\mathcal{S}}\Qft(\Crect^{\perp})$, whose fibrant objects are fibrant objects in $\Qft(\Crect^{\perp})$ that satisfy the homotopy $\Srect$-time-slice axiom.
	\end{prop}
	\begin{proof} The existence of  left Bousfield localization is ensured by \cite[Theorem 4.7]{barwick_left_2010} since the projective model structure on operadic algebras valued in $\Ch_{\mathbb{k}}$ is left proper; this fact is proved in Proposition \ref{prop_ExtensionModelIsLeftProper}. The characterization of fibrant objects is the content of Proposition \ref{prop_CharacterizationOfTIMESLICELOCALObjects}.
	\end{proof}
	    
	A comparison with the model category $\Qft_{\h\Srect}(\Crect^{\perp})$ (Proposition \ref{prop_STRUCTUREDweakTimeSliceAQFTModel}) is now in order.
	
	\begin{thm}\label{thm_ComparisonWeakTimeSliceModels}
		The homotopical localization morphism $\ell\colon \Op_{\Crect}^{\perp}\to\mathfrak{L}_{\Srect}\Op_{\Crect}^{\perp}$ induces a Quillen equivalence 
		$
		\mathsf{L}_{\mathcal{S}}\Qft(\Crect^{\perp})\rightleftarrows\Qft_{\h\Srect}(\Crect^{\perp}).
		$
	\end{thm} 
	    \begin{proof}
	    	We have to show that the Quillen pair $\ell_{\sharp}\colon\Qft(\Crect^{\perp})\rightleftarrows\Qft_{\h\Srect}(\Crect^{\perp}):\ell^*$ induced by $\ell$ descends to the left Bousfield localization and that this descent gives rise to an equivalence at the level of homotopy categories. 
	    	
	    	The universal property of the left Bousfield localization proves the first claim since $\Srect$ goes to equivalences in $\mathfrak{L}_{\Srect}\Op_{\Crect}^{\perp}$ via $\ell$. Alternatively, one can show that $\ell^*$ preserves fibrations between fibrant objects and that suffices.
	    	
	    	To check that the Quillen pair establishes an equivalence for homotopy categories, use Propositions \ref{prop_STRUCTUREDweakTimeSliceAQFTModel} and \ref{prop_WeakTimeSliceModel} to obtain the chain of equivalences
	    	$$
	    	    \Ho\big(\mathsf{L}_{\mathcal{S}}\Qft(\Crect^{\perp})\big)\simeq \Ho\big(\mathsf{L}_{\mathcal{S}}\Qft(\Crect^{\perp})_{\text{fibrant}}\big)\simeq \Ho\big(\QFT_{\h\Srect}(\Crect^{\perp})\big)\simeq \Ho\big(\Qft_{\h\Srect}(\Crect^{\perp})\big).
	    	$$
	    \end{proof}
    
    \begin{ex}\label{ex_ChiralFTsToOrientedTFTs}
    Let $\mathsf{Mfld}^{\mathsf{or}}_{m}$ be the category of oriented $m$-manifolds with orientation preserving open embeddings equipped with the disjoint orthogonality relation, i.e.\ $\upf\perp \upg$ if and only if $\im\upf\cap\im\upg= \cancel{\textup{o}}$. In this case,  $\QFT(\mathsf{Mfld}^{\mathsf{or},\perp}_{m})$ is the category of chiral conformal quantum field theories in dimension $m$ (see \cite[Example 8.4.7]{yau_homotopical_2019}). Taking $\Srect$ to be the set of isotopy equivalences, i.e.\ those embeddings that are isotopic to diffeomorphisms, one obtains the Quillen equivalent model structures
    $
    \Lrect_{\mathcal{S}}\Qft(\mathsf{Mfld}_{m}^{\mathsf{or},\perp}) 
    $
    and
    $  \Qft_{\h\Srect}(\mathsf{Mfld}_{m}^{\mathsf{or},\perp})
    $
    presenting the homotopy theory of chiral conformal quantum field theories sending isotopy equivalences to quasi-isomorphisms. Thus, those quantum field theories are topological. In fact, $\Lrect_{\mathcal{S}}\Qft(\mathsf{Mfld}_{m}^{\mathsf{or},\perp})$ being a left Bousfield localization of $\Qft(\mathsf{Mfld}_{m}^{\mathsf{or},\perp})$ tells us that any chiral conformal quantum field theory has a canonical ``topological field theory" associated to it in this sense, given by fibrantly replacing the theory in $\Lrect_{\mathcal{S}}\Qft(\mathsf{Mfld}_{m}^{\mathsf{or},\perp})$. 
    
     It might be interesting to compare $\Lrect_{\mathcal{S}}\Qft(\mathsf{Mfld}_{m}^{\mathsf{or},\perp})$ with an alternative axiomatization of oriented topological quantum field theories. At the moment, we do not know how to do that.
    \end{ex}

    \paragraph{Strict versus homotopy time-slice axiom.}
    We close this section by explaining how Theorem \ref{thm_ComparisonWeakTimeSliceModels} answers \cite[Open Problem 4.14]{benini_higher_2019} and how such a problem is connected to the \emph{strictification} of the homotopy time-slice axiom. First, let us recollect the ingredients and the content of this problem.	
    	
    In \cite[Example 4.13]{benini_higher_2019}, the authors introduce a new concept in the theory of algebraic field theories that they called  \emph{homotopy} $\Srect$-\emph{constancy}, but that we will refer to as \emph{quasi-strict}  $\Srect$-\emph{constancy} by reasons that will be clear later.  
    To recall their notion, let us denote $\mathcal{L}\colon \Op_{\Crect}^{\perp}\to \Op_{\Crect_{\Srect}}^{\perp_{\Srect}}$ the canonical map of operads and 
    $$
    \mathbb{L}\mathcal{L}_{\sharp}\colon\Qft(\Crect^{\perp})\rightleftarrows\Qft_{\Srect}(\Crect^{\perp})\colon\mathcal{L}^*
    $$
    its associated derived adjoint pair. 
    \begin{defn}\label{defn_HomotopyConstancy}
    A field theory $\Aalg\in\QFT(\Crect^{\perp})$ is \emph{quasi-strict} $\Srect$-\emph{constant} if the canonical map $\Aalg\to\mathcal{L}^*\mathbb{L}\mathcal{L}_{\sharp}\Aalg$ is an equivalence.
    \end{defn}

     In particular, if $\Aalg$ is \emph{quasi-strict} $\Srect$-\emph{constant}, it is equivalent to a \emph{strictly} $\Srect$\emph{-constant} AQFT.
    
    Benini and Schenkel wondered if quasi-strict $\Srect$-constancy captures homotopy $\Srect$-constancy in our sense, as discussed at the end of \cite[Section 3.4]{benini_higher_2019}. However, it is not even clear if a strict $\Srect$-constant field theory (Definition  \ref{defn_QFTsSatisfyingTimeSliceAxiom}) is quasi-strict $\Srect$-constant, whereas strict $\Srect$-constancy clearly implies homotopy $\Srect$-constancy. The main obstruction for this to happen is that the derived counit $\mathbb{L}\mathcal{L}_{\sharp}\mathcal{L}^*\Rightarrow\id$ is not known to be a natural equivalence.
    Hence, we can reformulate \cite[Open Problem 4.14]{benini_higher_2019} as the following list of items:
    \begin{enumerate} 
    	\item[(Q1)] Is quasi-strict $\Srect$-constancy equivalent to homotopy $\Srect$-constancy?
    	\item[(Q2)] Does strict $\Srect$-constancy imply quasi-strict $\Srect$-constancy?
    	\item[(Q3)] Is the derived counit  $\mathbb{L}\mathcal{L}_{\sharp}\mathcal{L}^*\Rightarrow\id$ a natural equivalence?
    \end{enumerate}
    
    We solve these questions by considering homotopy $\Srect$-constancy instead of quasi-strict $\Srect$-constancy (which is actually the choice that one is pursuing at the end), by virtue of the results discussed in this section. Denoting by $\ell\colon\Op_{\Crect}^{\perp}\to \mathfrak{L}_{\Srect}\Op_{\Crect}^{\perp}$ the homotopical localization map and  
    $$
     \mathbb{L}\ell_{\sharp}\colon\Qft(\Crect^{\perp})\rightleftarrows \Qft_{\h\Srect}(\Crect^{\perp})\colon\ell^*
    $$
    its associated derived adjoint pair, we now explain how all the complications vanish if one replaces $\mathbb{L}\mathcal{L}_{\sharp}\dashv\mathcal{L}^*$ by $\mathbb{L}\ell_{\sharp}\dashv\ell^*$. Note that the homotopy category $\Ho\QFT_{\h\Srect}(\Crect^{\perp})$ is by Theorem \ref{thm_ComparisonWeakTimeSliceModels} a full reflective subcategory of $\Ho\QFT(\Crect^{\perp})$ with reflector $\mathbb{L}\ell_{\sharp}$, since $\Ho\QFT_{\h\Srect}(\Crect^{\perp})$ is presented in Proposition \ref{prop_WeakTimeSliceModel} as a left Bousfield localization of $\Qft(\Crect^{\perp})$. This fact may be reinterpreted as: (i) the derived counit $\mathbb{L}\ell_{\sharp}\ell^*\Rightarrow\id$ is an  equivalence and (ii) $\Aalg\in\Qft(\Crect^{\perp})$ is homotopy $\Srect$-constant iff $\Aalg\to\ell^*\mathbb{L}\ell_{\sharp}\Aalg$ is an equivalence. Moreover, strictly $\Srect$-constant field theories are canonically examples of homotopy $\Srect$-constant field theories.
    
    Now we compare quasi-strict $\Srect$-constancy and homotopy $\Srect$-constancy. By the universal property of $\ell$, there is a canonical factorization of $\mathcal{L}$ through $\ell$, $\mathcal{L}=\upxi\cdot\ell$, which yields a factorization of the right Quillen functor $\mathcal{L}^*$
    $$
    \Qft(\Crect^{\perp})\overset{\ell^*}{\longleftarrow} \Qft_{\h\Srect}(\Crect^{\perp})\overset{\upxi^*}{\longleftarrow} \Qft_{\Srect}(\Crect^{\perp})\colon \mathcal{L}^*.
    $$
    The best scenario would be that the right map $\Qft_{\h\Srect}(\Crect^{\perp})\leftarrow\Qft_{\Srect}(\Crect^{\perp}):\upxi^*$ were part of a Quillen equivalence, since this is equivalent to the statement (Q1): quasi-strict $\Srect$-constancy is equivalent to homotopy $\Srect$-constancy (because both are properties stable under equivalences for algebraic field theories in $\QFT(\Crect^{\perp})$). Thus, the original question (Q1) can be reformulated as a \emph{strictification problem} from homotopy $\Srect$-constancy to strict $\Srect$-constancy. However, such a result in general seems unlikely, as the following examples show.
    
   \begin{ex}\label{ex_HomotopySconstantTQFT} Let us provide an example that is motivated by topological field theories for which the answer to (Q1) is negative, i.e.\ there are AQFTs satisfying the homotopy $\Srect$-time-slice axiom that are not quasi-strict $\Srect$-constant AQFTs.
   \end{ex}
   	
   Choose a smooth $n$-manifold $\Xrect$ and consider the orthogonal category $\Disc(\Xrect)^{\perp}$ given by:
   \begin{itemize}
   	\item $\Disc(\Xrect)$ is the poset of open $n$-discs $\Urect\subseteq\Xrect$ and inclusions between them;
   	\item $\perp$ is the disjointness orthogonality relation, i.e.\ $(\Urect_1\subseteq\Urect)\perp(\Urect_2\subseteq\Urect)$ iff $\Urect_1\cap\Urect_2=\cancel{\textup{o}}$.
   \end{itemize}
   	Take $\Srect$ to be the set of isotopy equivalences. Then, \cite[Subsection 2.4]{ayala_factorization_2015} implies that the homotopy localization  $\mathfrak{L}_{\Srect}\Disc(\Xrect)$ of the category $\Disc(\Xrect)$ is the following topological category:
    \begin{itemize}
    	\item its objects are open embeddings $\upnu\colon \mathbb{R}^{n}\hookrightarrow \Xrect$;
    	\item its hom-space $\Map(\upnu,\upnu')$ is the homotopy fiber of
    	$
    	\upnu'_*\colon\mathsf{Emb}(\mathbb{R}^n,\mathbb{R}^n)\rightarrow\mathsf{Emb}(\mathbb{R}^n,\Xrect)
    	$
    	at $\upnu\colon\mathbb{R}^n\hookrightarrow \Xrect$,
    	where $\mathsf{Emb}$ denotes the space of smooth embeddings.
    \end{itemize}
    By \cite[Proposition 6.4]{horel_factorization_2017}, the map $\upnu'_*$ above can be identified with the map induced by the derivative of $\upnu'$
    $$
    \mathsf{O}(n)\simeq \mathsf{O}(n)\times \mathbb{R}^{n}\cong \mathsf{Fr}(\mathbb{R}^n)\overset{\T\upnu'}{\longrightarrow} \mathsf{Fr}(\Xrect),
    $$
    where the left hand side is the $n^{\text{th}}$-orthogonal group and the right hand side is the $\mathsf{O}(n)$-frame bundle associated to the tangent bundle $\T\Xrect$.
     Hence, there are manifolds $\Xrect$ for which the hom-spaces  $\mathsf{Map}(\upnu,\upnu')$ will have non-trivial homotopy groups, e.g.\ any parallelizable manifold with non-trivial homotopy groups due to the long exact sequence associated to a fiber sequence. As a consequence of this fact, we have:
     
     \begin{prop}\label{prop:NonstricficableTAQFTs}
     	There are AQFTs over $\Disc(\Xrect)^{\perp}$ satisfying the homotopy $\Srect$-time-slice axiom which cannot be equivalent to any AQFT satisfying the strict $\Srect$-time-slice axiom (in particular, they are not quasi-strict $\Srect$-constant). Therefore, the  Quillen adjunction 
     	$$
     	\upxi_{\sharp}\colon\Qft_{\h\Srect}(\Disc(\Xrect)^{\perp})\rightleftarrows \Qft_{\Srect}(\Disc(\Xrect)^{\perp}):\upxi^*
     	$$
     	is not a Quillen equivalence.
     \end{prop}
     \begin{proof} We first observe that an AQFT $\Balg$ satisfying the $\Srect$-strict-time slice axiom admits an (unique) action of $\mathfrak{L}_{\Srect}\Disc(\Xrect)$ by neglection of structure and the universal property of the homotopy localization, but this action must factor by the canonical action of the ordinary localization $\Disc(\Xrect)_{\Srect}$ on $\Balg$. Hence, any 
     AQFT $\Balg'$ over $\Disc(\Xrect)^{\perp}$ equivalent to $\Balg$ inherits a homotopy coherent action of $\mathfrak{L}_{\Srect}\Disc(\Xrect)$ factoring by a homotopy coherent action of $\Disc(\Xrect)_{\Srect}$. By the universal property of the homotopy localization, this action coincides (up to equivalence) with the unique $\mathfrak{L}_{\Srect}\Disc(\Xrect)$-action on $\Balg'$ coming from the fact that $\Balg'$ sends maps in $\Srect$ to quasi-isomorphisms. Since $\Disc(\Xrect)_{\Srect}$ has discrete hom-spaces (in fact, the connected components of hom-spaces of $\mathfrak{L}_{\Srect}\Disc(\Xrect)$ \cite[Corollary 4.2]{dwyer_function_1980}), any homotopy $\Srect$-constant AQFT whose underlying action over  $\mathfrak{L}_{\Srect}\Disc(\Xrect)$ does not only depend on the connected components of the hom-spaces $\Map(\upnu,\upnu')$ cannot be equivalent to a strictly $\Srect$-constant AQFT.
     
     Let us construct an example of such homotopy $\Srect$-constant AQFTs. Fix an open embedding $\upnu\colon \mathbb{R}^{n}\hookrightarrow \Xrect$ and its associated representable functor with values in spaces  $\Map(\upnu,\text{-})$. Taking $\mathbb{k}$-chains and using the symmetric algebra functor, we end up with 
     $$
     \Aalg\colon \Disc(\Xrect)\longrightarrow \mathsf{cdga}_{\mathbb{k}}, \quad \Urect\mapsto \mathsf{Sym}\big(\Crect_{*}(\Map(\upnu,\mathsf{inc}_{\Urect});\mathbb{k})\big),
     $$   
     where $\mathsf{cdga}_{\mathbb{k}}$ is the category of commutative differential graded $\mathbb{k}$-algebras and we have denoted by $\mathsf{inc}_{\Urect}\colon \Urect\subseteq\Xrect$ the open inclusion of $\Urect$ into $\Xrect$.	One can readily check that $\Aalg$ belongs to $\QFT_{\h\Srect}(\Disc(\Xrect)^{\perp})$. If there is some homotopy group of $\Map(\upnu,\mathsf{inc}_{\Urect})$ which is not completely torsion, $\Aalg$ satisfies the claim.

     
     
     \end{proof}
    
    Moreover, we provide a completely homotopical example that shows how remote a positive answer to (Q1) is by comparing how different $\mathfrak{L}_{\Srect}\Crect$ and $\Crect_{\Srect}$ can be in general.
    
\begin{ex}\label{ex_HomotopyConstancy} Let $\Crect^{\perp_0}$ be an orthogonal category with the initial orthogonality relation, i.e.\ the one for which there are no orthogonal maps. The category of algebraic quantum field theories $\QFT(\Crect^{\perp_0})$ is the category $\Fun(\Crect,\Mon(\V))$ of $\Crect$-diagrams of monoids. Then, for a set of maps $\Srect$ in $\Crect$, the canonical Quillen functor $\Qft_{\textup{w}\Srect}(\Crect^{\perp_0})\leftarrow\Qft_{\Srect}(\Crect^{\perp_0})$ yields a Quillen equivalence if and only if the canonical restriction functor
\begin{equation}\label{eqt_RestrictionBetweenWeakStrictConstanFieldTheories}
\Ho
\Fun(
\mathfrak{L}_{\Srect}\Crect, \Mon(\V))
\longleftarrow	
\Ho\Fun(
\Crect_{\Srect}, \Mon(\V))
\end{equation}
is an equivalence of categories. We note that it is quite unlikely that the functor (\ref{eqt_RestrictionBetweenWeakStrictConstanFieldTheories}) is in general an equivalence since $\mathfrak{L}_{\Srect}\Crect$ and $\Crect_{\Srect}$ can be potentially quite different.

Consider that $\V$ is the category of simplicial sets $\sSet$ with the Kan-Quillen model structure. Then, on the one hand, the ordinary localization $\Crect_{\Srect}$ coincides with the category $\pi_0\mathfrak{L}_{\Srect}\Crect$, obtained from $\mathfrak{L}_{\Srect}\Crect$ by taking connected components of its mapping spaces \cite[Corollary 4.2]{dwyer_function_1980}. On the other hand, any small simplicial category (or $\infty$-category) can be obtained, up to Dwyer-Kan equivalence, as $\mathfrak{L}_{\Srect}\Crect$ for a suitable choice of $(\Crect,\Srect)$ due to \cite{barwick_relative_2012}. Thus, one can take $(\Crect,\Srect)$ in such a way that $\mathfrak{L}_{\Srect}\Crect$ is a model for the classifying $\infty$-groupoid $\mathsf{BG}$ of a simply connected topological group $\mathsf{G}$. Under these conditions, the functor (\ref{eqt_RestrictionBetweenWeakStrictConstanFieldTheories}) becomes the diagonal functor
 $
\Ho
\Fun(
\mathsf{BG}, \Mon(\V))
\leftarrow	
\Ho \Mon(\V),
 $
which sends a monoid to the constant functor at that monoid. If we further consider that $\sSet$ is endowed with the cocartesian monoidal structure, the category of monoids is just $\sSet$ itself. The diagonal functor in this case can be identified with the composite
  $$
 \Ho\sSet_{/\mathsf{BG}}
 \overset{\sim}{\longleftarrow}
 \Ho
 \Fun(
 \mathsf{BG},\sSet)
 \longleftarrow	
 \Ho \sSet,
 $$
 defined by sending a space $\X$ to the second factor projection  $\X\times\mathsf{BG}\to\mathsf{BG}$, as a consequence of \cite[Section 8]{shulman_parametrized_2008} (or the references cited therein for the simplicial case). This functor cannot be an equivalence for a general $\mathsf{G}$; that is, the answer to (Q1) is negative. For instance, the universal fibration $\mathsf{EG}\to\mathsf{BG}$ does not belong to the essential image of this functor, since this would imply that $\mathsf{H}^*(\mathsf{BG})$ can be embedded into $\mathsf{H}^*(\mathsf{EG})\cong \mathsf{H}^*(\pt)$. Moreover, this functor is not fully faithful in general. One way to see it is by computing the (homotopy) mapping space between images of the functor in $\sSet_{/\mathsf{BG}}$ to see that 
 $$
 \Map_{\mathsf{BG}}(\X\times\mathsf{BG},\Y\times\mathsf{BG})\simeq \Map(\X\times\mathsf{BG},\Y).
 $$
 Taking connected components, the map 
 $$
 \Ho\sSet(\X,\Y)\longrightarrow \pi_{0}\Map(\X\times\mathsf{BG},\Y)\cong\Ho\sSet_{/\mathsf{BG}}(\X\times\mathsf{BG},\Y\times\mathsf{BG})
 $$
 is not a bijection in general. In fact, for $\X$ contractible, $\mathsf{G}$ the Eilenberg-MacLane space $\mathsf{K}(\mathbb{Z},2)$ and $\Y=\mathsf{BG}$, taking connected components we get $\pi_{0}(\Y)\cong 0$ and $\mathsf{H}^3(\mathsf{K}(\mathbb{Z},3),\mathbb{Z})\cong\mathbb{Z}.$ 
\end{ex}

 Note that the failure of fully-faithfulness above states that  
$$
\upxi^*\colon\Ho\QFT_{\Srect}(\Crect^{\perp})\longrightarrow \Ho\QFT_{\h\Srect}(\Crect^{\perp}) 
$$
is not fully faithful,  or equivalently that the derived counit $\mathbb{L}\upxi_{\sharp}\upxi^*\Rightarrow\id$ is not a natural equivalence, where $\upxi$ is the remaining factor in the factorization $\mathcal{L}=\upxi\cdot \ell$. Observe that one can write the derived counit in (Q3) as the following composition
$$
\begin{tikzcd}[ampersand replacement=\&]
\mathbb{L}\mathcal{L}_{\sharp}\mathcal{L}^*=\mathbb{L}\upxi_{\sharp}\mathbb{L}\ell_{\sharp}\ell^*\upxi^*\ar[rr, Rightarrow] \&\& \mathbb{L}\upxi_{\sharp}\upxi^*\ar[rr, Rightarrow] \&\& \id.
\end{tikzcd}
$$
The natural transformation on the left is a natural equivalence (apply $\mathbb{L}\upxi_{\sharp}$ to the natural equivalence in (i) above), and thus $\mathbb{L}\mathcal{L}_{\sharp}\mathcal{L}^*\Rightarrow\id $ is a natural equivalence iff $\mathbb{L}\upxi_{\sharp}\upxi^*\Rightarrow\id$ is so. Therefore, (Q3) is negative in general.


\begin{rem}\label{rem_StrictificationPAPER} The discussion above points towards a main obstruction for homotopy $\Srect$-constancy to coincide with (quasi-)strict $\Srect$-constancy in general. However, physical examples may be blind to this issue for several reasons. For instance, in \cite{bruinsma_relative_2021}, the authors show that linear homotopy AQFTs admit a strictification for the homotopy time-slice axiom when $\Crect^{\perp}$ is the relative Cauchy evolution category, i.e.\ it exists a weakly equivalent AQFT satisfying the strict time-slice axiom. The recent preprint \cite{benini_strictification_2022} addresses this kind of question in much generality and it provides stronger results with simpler proofs (due to all the constructions developed in the present work). In particular, we prove that 
	$
	\Qft_{\h\Srect}(\Crect^{\perp})\leftarrow\Qft_{\Srect}(\Crect^{\perp})
	$
is a Quillen equivalence for: (i) Haag-Kastler-type AQFTs on a fixed globally hyperbolic Lorentzian manifold (with or without time-like boundary), (ii) locally covariant AQFTs in one spacetime dimension, (iii) locally covariant conformal AQFTs in two spacetime dimensions and (iv) general AQFTs when $\Crect^{\perp}$ is the relative Cauchy evolution category. That is, the homotopy time-slice axiom can be strictified in those situations. In fact, it seems to be a general phenomenon that AQFTs in Lorentzian settings could admit strictifications for the time-slice axiom.  
\end{rem}

\end{subsection}
\end{section}

\begin{section}{Extension model structure}\label{sect_ExtensionModelStructure}
	
Our goal is to produce the extension model structure for operadic algebras, which, aside its intrinsic interest, will be crucial in the main results of next section. Given an "inclusion" of operads, the formal idea of the construction consists of finding a cellularization, also known as right Bousfield localization, of the projective model structure on algebras for the bigger operad, which is Quillen equivalent to the projective model structure for the smaller one.
 
 The discussion in this section is more technical, and related to the theory of Quillen model stuctures. Therefore, it could be skipped in a first lecture and consulted when referred in the sequel.
 
 Let us fix a closed symmetric monoidal model category $\V$ (see \cite{fresse_homotopy_2017-1}). We think about  $\V$ as a sufficiently structured homotopical cosmos in which our constructions hold. Several hypotheses must be satisfied for the extension model structure to exist, and they will be stated in due time, when needed.
 
 \begin{notat}
 	We say that a functor, between relative  categories \cite{barwick_relative_2012}, is homotopical if it preserves equivalences.
 \end{notat}
 
 Let $\upiota\colon \OpB\to \OpN$ be a morphism of $\V$-operads. Then, the induced restriction functor $\upiota^*$ between algebras admits a left adjoint
 $$
 \upiota_{\sharp}\colon \Alg_{\OpB}(\V)\rightleftarrows \Alg_{\OpN}(\V)\colon \upiota^*.
 $$
 We want to deal with homotopy theories on these categories, so we assume:
 \begin{hyp}\label{hyp_VNiceEnough}
 	$\V$ is cofibrantly generated \cite[Definition 11.1.2]{hirschhorn_model_2003} and $\OpB$, $\OpN$ are admissible operads, i.e.\ the categories $\Alg_{\OpB}(\V)$ and $\Alg_{\OpN}(\V)$ admit the projective model structure.
 \end{hyp}
 
Endowing both categories with the projective model structure, the adjunction $\upiota_{\sharp}\dashv\upiota^*$ is a Quillen pair with $\upiota^*$ being a homotopical functor. With this observation, we can restate our goal in this section as the search of a sufficiently broad condition on $\upiota\colon\OpB\to\OpN$ such that the induced derived adjunction
$$
\mathbb{L}\upiota_{\sharp}\colon \Ho\Alg_{\OpB}(\V)\rightleftarrows\Ho\Alg_{\OpN}(\V)\colon \upiota^*
$$
identifies $\Ho\Alg_{\OpB}(\V)$ with a full coreflective subcategory of $\Ho\Alg_{\OpN}(\V)$, which moreover can be modelled by what we call the extension model structure on $\Alg_{\OpN}(\V)$. We achieve this intent in two steps: firstly, we provide a model structure on $\Alg_{\OpN}(\V)$ which exists in great generality, but which does not always produce what we expect; secondly, we introduce a condition on $\upiota\colon\OpB\to\OpN$ and we check that our model category  fulfills our plan assuming it.

\begin{prop}\label{prop_ExtensionModelCofibrantlyGenerated} The projective model structure on $\Alg_{\OpN}(\V)$ admits a cellularization (or right Bousfield localization) called extension model structure with the following properties:
	\begin{itemize}
		\item A map $\upf\colon\Aalg\to\Balg$ is a weak equivalence iff $\upf_{\upiota\textup{b}}\colon \Aalg(\upiota\upb)\to\Balg(\upiota\upb)$ is an equivalence for any $\upb\in\col(\Balg)$, i.e. $\upiota^*\upf$ is an equivalence.
	    \item The class of fibrations in the extension model structure equals that of proj-fibrations.
	    \item It is cofibrantly generated (resp. combinatorial if $\V$ is so).
	\end{itemize}
\end{prop}
\begin{proof} First note that if it exists, the extension model structure is a cellularization of the projective model structure since its class of weak equivalences contains that of proj-equivalences. The existence of this cellularization will be an application of Kan's recognition theorem \cite[Theorem 11.3.1]{hirschhorn_model_2003} of model categories. Thus, we must find candidates of generating (trivial) cofibrations that fulfill the requirements in loc.\ cit. 
	
	Since the fibrations of the extension model structure are the proj-fibrations, it suffices to take as generating acyclic cofibrations the set $\textup{J}$ of generating acyclic proj-cofibrations. 
	
	In order to find a set of generating cofibrations, note that they should detect trivial fibrations, which are proj-fibrations that are also equivalences when restricted to $\col(\OpB)$ along $\upiota$. The proj-fibration condition is fulfilled if we simply add $\textup{J}$ to our generating cofibrations. Being additionally an equivalence in the above sense is detected by right lifting property against a set $\textup{I}$. Such a set $\textup{I}$ is just the image along the composite functor $\upiota_{\sharp}\Frect_{\OpB}$ of the set of generating cofibrations in $\V^{\times \col(\OpB)}$ due to adjunction, where $\Frect_{\OpB}$ denotes the free $\OpB$-algebra functor. Note that the composite functor $\upiota_{\sharp}\Frect_{\OpB}$ fits into the commutative square of left adjoint functors
	$$
	\begin{tikzcd}[ampersand replacement=\&]
    \V^{\times \col(\OpB)} \ar[r,"\col(\upiota)_!"]\ar[d,"\Frect_{\OpB}"']\& \V^{\times \col(\OpN)}\ar[d, "\Frect_{\OpN}"]\\
	\Alg_{\OpB}(\V)\ar[r,"\upiota_{\sharp}"']\&\Alg_{\OpN}(\V)
	\end{tikzcd}\quad .
	$$ 
	
	Summarizing, $(\textup{I}\sqcup\textup{J},\textup{J})$ are the candidates for generating (acyclic) cofibrations of the extension model structure. So, we must check if they fulfill the requirements in \cite[Theorem 11.3.1]{hirschhorn_model_2003}. Both sets admit the small object argument, by the same reasons given to ensure the existence of the projective model structure on $\Alg_{\OpN}(\V)$. Hence, we are reduced to check that $\textup{I}\sqcup\textup{J}$-fibrations (also called $\textup{I}\sqcup\textup{J}$-injective maps) are proj-fibrations and equivalences in the extension model structure and that $\textup{J}$-cofibrations are equivalences in the extension model structure. Both facts are easy consequences of the definitions. 
	
    What is missing to conclude the proof is checking that the extension model structure is combinatorial when $\V$ is so. This follows from the fact that $\V$ being presentable implies that $\Alg_{\OpN}(\V)$ is presentable.
\end{proof}

In general, it is not clear if the adjunction $\upiota_{\sharp}\colon \Alg_{\OpB}(\V)\rightleftarrows \Alg_{\OpN}(\V)\colon \upiota^*$ yields a Quillen equivalence between the extension model structure on the right and the projective model structure on the left; this property is the one that we want at the end. However, we will prove that the only additional hypothesis for this to hold is the following.
 \begin{hyp}\label{hyp_OperadicLanHomotopicallyFF}
	The unit $\id\to\upiota^*\upiota_{\sharp}$ is an equivalence on proj-cofibrant algebras\footnote{We interpret this condition as $\mathbb{L}\upiota_{\sharp}$ being homotopically fully faithful. It is ensured if $\upiota$ is fully faithful in the ordinary operadic sense.}.
\end{hyp}
 The proof consists of identifying more precisely cofibrant objects and cofibrations between cofibrant objects in the extension model structure. We do so by constructing part of the extension model structure by other means. In particular, we apply the dual of  \cite[Theorem 3.6]{carmona_Bousfield_2022} to the projective model structure on $\Alg_{\OpN}(\V)$ and a particular augmented endofunctor $(\Qrep,\upepsilon)$ that we now define.
 
  Fix a functorial cofibrant replacement $(\Q,\q)$ on $\Alg_{\OpB}(\V)$. Then, one gets the derived pair between projective model categories 
 $$
 \upiota_{\sharp}\Q\colon \Alg_{\OpB}(\V)\rightleftarrows\Alg_{\OpN}(\V)\colon \upiota^*,
 $$
 and the composite endofunctor $\Qrep=\upiota_{\sharp}\Q\upiota^*$ on $\Alg_{\OpN}(\V)$ admits an augmentation that should be seen as the derived counit of the derived pair,
 $$
 \begin{tikzcd}[ampersand replacement= \&]
 \upepsilon\colon \Qrep=\upiota_{\sharp}\Q\upiota^*\ar[r, "\upiota_{\sharp}\q\upiota^*"]\& \upiota_{\sharp}\upiota^*\ar[r, "\text{counit}"]\&\id.
 \end{tikzcd}
 $$
 
 In order to apply loc.\ cit.\ to $(\Qrep,\upepsilon)$ and relate its consequent homotopical structure with the extension model structure of Proposition \ref{prop_ExtensionModelCofibrantlyGenerated}, one is reduced to check the following easy results.
 
 \begin{lem}\label{lem_QrepIsCellularizator} Assume Hypothesis \ref{hyp_OperadicLanHomotopicallyFF} holds. Then, the augmented endofunctor $(\Qrep,\upepsilon)$ satisfies: $\Qrep$ is homotopical endofunctor and both $\upepsilon_{\Qrep}, \,\Qrep\upepsilon$ are natural weak equivalences.
 \end{lem}
 \begin{proof}
 	It is clear that $\Qrep$ is homotopical since it is a composition of homotopical functors. Under Hypothesis \ref{hyp_OperadicLanHomotopicallyFF}, let us see that $\upepsilon_{\Qrep}$ is an equivalence. We deduce this fact from the commutative diagram
 	$$
 	\begin{tikzcd}[ampersand replacement=\&]
 	\upiota_{\sharp}\Q\upiota^*\upiota_{\sharp}\Q\upiota^*\ar[rr,"\upiota_{\sharp}\q\upiota^*\upiota_{\sharp}\Q\upiota^*"]\ar[rrrr,"\upepsilon_{\Qrep}", bend left=30] \&\& \upiota_{\sharp}\upiota^*\upiota_{\sharp}\Q\upiota^*\ar[rr,"\text{counit}\cdot\upiota_{\sharp}\Q\upiota^*"] \&\& \upiota_{\sharp}\Q\upiota^*.\\\\
 	\upiota_{\sharp}\Q\Q\upiota^*\ar[rr,"\sim" , "\upiota_{\sharp}\q\Q\upiota^*"']\ar[uu,"\wr"', "\upiota_{\sharp}\Q\cdot\text{unit}\cdot\Q\upiota^*"] \&\& \upiota_{\sharp}\Q\upiota^* \ar[uu, "\upiota_{\sharp}\cdot\text{unit}\cdot\Q\upiota^*"'] \ar[rruu, equal, bend right=20]
 	\end{tikzcd} 
 	$$ 
 	The commutativity comes from the naturality of the unit and the triangular identity of the adjunction $\upiota_{\sharp}\dashv\upiota^*$; the conclusion follows from 2-out of-3 for equivalences. The claim for $\Qrep\upepsilon$ follows from a similar diagram chasing.
 \end{proof}
\begin{rem} In the notation of  \cite{carmona_Bousfield_2022}, Lemma \ref{lem_QrepIsCellularizator} shows that the augmented endofunctor $(\Qrep,\upepsilon)$ is a BF-coreflector (see \cite[Definition 3.3]{carmona_Bousfield_2022}). This fact will allow us to apply the dual of \cite[Theorem 3.6]{carmona_Bousfield_2022} to identify cofibrant objects of the extension model structure in Theorem \ref{thm_ExistenceOfExtensionModelAndProperties}.
\end{rem}

 \begin{lem}\label{lem_ColocalEquivalences} Assume Hypothesis \ref{hyp_OperadicLanHomotopicallyFF} holds. Then,
 	the class of maps in $\Alg_{\OpN}(\V)$
 	$$
 	\left\{\upf\colon\Aalg\to \Balg\text{ such that }\Qrep\upf\colon \Qrep\Aalg\xrightarrow{}\Qrep\Balg\text{ is an equivalence}\right\}
 	$$
 	coincides with the weak equivalences in the extension model structure (Proposition \ref{prop_ExtensionModelCofibrantlyGenerated}).
 \end{lem}
 \begin{proof}
 	Let us assume that $\upf$ is a weak equivalence in the extension model structure. Then, $\Qrep\textup{f}$ is an equivalence since it is obtained from the equivalence $\upiota^*\textup{f}$ via the homotopical functor $\upiota_{\sharp}\Q$.
 	
 	Conversely, assume that $\upf$ is such that $\Qrep\textup{f}$ is an equivalence. Then we have a commutative diagram
 	$$ 
 	\begin{tikzcd}[ampersand replacement=\&]
 	\upiota^*\Qrep\Aalg\ar[r,equal] \ar[d,"\upiota^*\Qrep\textup{f}"',"\wr"]\& \upiota^*\upiota_{\sharp}\Q\upiota^*\Aalg \ar[d,"\upiota^*\upiota_{\sharp}\Q\upiota^*\textup{f}"']\& \Q\upiota^*\Aalg\ar[l, "\text{unit}"',"\sim"] \ar[r,"\q","\sim"']\ar[d,"\Q\upiota^*\textup{f}"']\& \upiota^*\Aalg \ar[d,"\upiota^*\textup{f}"]\\
 	\upiota^*\Qrep\Balg\ar[r,equal] \& \upiota^*\upiota_{\sharp}\Q\upiota^*\Balg \& \Q\upiota^*\Balg\ar[l, "\text{unit}", "\sim"'] \ar[r,"\q"',"\sim"]\& \upiota^*\Balg
 	\end{tikzcd}
 	$$
 	which by 2-out of-3, Hypothesis \ref{hyp_OperadicLanHomotopicallyFF} and the fact that $\upiota^*$ is homotopical allows us to conclude the result. 
 \end{proof}
 
 These preparations lead to the fundamental result in this section. But first and for future reference, we introduce the following notion.
 
 \begin{defn}\label{defn:ColocalAlgebra}
  We say that an $\OpN$-algebra is $\OpB$-\emph{colocal} if the augmentation $\upepsilon$ on it is an equivalence, i.e. $\Aalg$ is $\OpB$-colocal if $\Qrep\Aalg\xrightarrow{\sim}\Aalg$. 
 \end{defn}
  
 \begin{thm}\label{thm_ExistenceOfExtensionModelAndProperties}
  Assuming Hypothesis \ref{hyp_OperadicLanHomotopicallyFF},	the extension model structure on $\Alg_{\OpN}(\V)$ constructed in Proposition \ref{prop_ExtensionModelCofibrantlyGenerated} satisfies:
  \begin{itemize}
 		\item A cofibrant object is a proj-cofibrant algebra which is $\OpB$-colocal.
 		\item A map between cofibrant objects is a cofibration iff it is a proj-cofibration.
 	\end{itemize}
 	Furthermore, the Quillen pair $\upiota_{\sharp}\dashv\upiota^*$ descends to a Quillen equivalence between the projective model structure on $\Alg_{\OpB}(\V)$ and the extension model structure on $\Alg_{\OpN}(\V)$.
 \end{thm}
 \begin{proof}	
 	On the one hand, the recognition of cofibrant objects and cofibrations between them follows from the dual of \cite[Theorem 3.6]{carmona_Bousfield_2022} applied to $(\Qrep,\upepsilon)$. The reason is that loc.\;cit.\;can be applied due to Lemma \ref{lem_QrepIsCellularizator}, yielding almost a model structure $\Alg_{\OpN}(\V)_{\Qrep}$ (see \cite[Definition 2.6]{carmona_Bousfield_2022}) which shares fibrations and equivalences with the extension model structure by definition and Lemma \ref{lem_ColocalEquivalences}.
 
 	On the other hand, the Quillen pair $\upiota_{\sharp}\dashv\upiota^*$ descends to a Quillen pair for the extension model structure since $\upiota^*$ preserves fibrations and weak equivalences. Moreover, by Hypothesis \ref{hyp_OperadicLanHomotopicallyFF} and construction, the derived unit and counit are equivalences, and so we have the desired Quillen equivalence.
 \end{proof}
 
 We close this section by showing that the extension model structure is left proper when $\V=\Ch_{\mathbb{k}}$ with $\mathbb{k}$ field of characteristic zero. This fact is fundamental to further localize the extension model structure using \cite[Theorem 4.7]{barwick_left_2010} or \cite[Theorem 4.1.1]{hirschhorn_model_2003}.
 
 \begin{prop} \label{prop_ExtensionModelIsLeftProper}
 	The projective and the extension model structures on $\Alg_{\OpN}(\Ch_{\mathbb{k}})$ are left proper for any operad $\OpN$.
 \end{prop}
\begin{proof} Both statements are showed by proving a slightly more general fact: equivalences in the extension model structure are stable under cobase changes along projective cofibrations, because our argument is completely adaptable to prove that projective equivalences are stable under cobase changes along projective cofibrations. Being more precise, we show that given a pushout square in $\Alg_{\OpN}(\V)$ 
	$$
	\begin{tikzcd}[ampersand replacement=\&]
		\Aalg\ar[r,"\upf"]\ar[d, "\upg"'] \ar[rd, phantom, "\ulcorner" very near start] \& \Balg\ar[d]\\
	\widetilde{\Aalg}\ar[r, "\widetilde{\upf}"'] \& \widetilde{\Balg}
	\end{tikzcd}
	$$	
 where $\upf$ is an equivalence in the extension model structure and $\upg$ a projective cofibration, then  $\widetilde{\upf}$ is an equivalence in the extension model structure.
 
 First note that, as equivalences in the extension model structure are closed under retracts, we can consider without loss of generality that $\upg$ is a cellular proj-cofibration, i.e. $\upg$ is a transfinite composite of pushouts of generating proj-cofibrations $\upg_{\upalpha}$. Recalling that the projective model structure on $\Alg_{\OpN}(\V)$ is transferred through the (free $\OpN$-algebra, forgetful) adjunction 
 $$
 \Frect_{\OpN}\colon \V^{\times\col(\OpN)}\rightleftarrows\Alg_{\OpN}(\V)\colon \textup{U},
 $$
 it is clear that each pushout in the transfinite composite is of the form
 $$
 \begin{tikzcd}[ampersand replacement=\&]
 \Frect_{\OpN}(\source(\upj))\ar[r]\ar[d, rightarrowtail, "\Frect_{\OpN}(\upj)"'] \ar[rd, phantom, "\ulcorner" very near start] \& \Aalg_{\upalpha}\ar[d,"\upg_{\upalpha}",rightarrowtail]\\
 \Frect_{\OpN}(\target(\upj))\ar[r] \& \Aalg_{\upalpha+1},
 \end{tikzcd}
 $$
 where $\upj$ is a generating cofibration of $\V$ concentrated in one color. Due to \cite[Proposition 4.3.17]{white_bousfield_2018}, $\upg_{\upalpha}$, when viewed in $\V^{\times\col(\OpN)}$ via the forgetful functor, can be described as a $\upomega$-transfinite composite of colorwise cofibrations:
 $$
 \begin{tikzcd}[ampersand replacement=\&]
 \upg_{\upalpha}\colon \Aalg_{\upalpha}=\Aalg_{\upalpha}^{0}\ar[r,"\upg_{\upalpha}^{\,1}",rightarrowtail]\&\Aalg_{\upalpha}^{1}\ar[r,"\upg_{\upalpha}^{\,2}",rightarrowtail]\&\cdots\ar[r,rightarrowtail]\&\Aalg_{\upalpha}^{\upomega}=\Aalg_{\upalpha+1}.
 \end{tikzcd}
 $$ 
 Arranging the two filtrations together, $\upg$ is described as a transfinite composite of cofibrations in $\V^{\times \col(\OpN)}$. Using this, the initial pushout square is decomposed into the following commutative diagram in $\V^{\times \col(\OpN)}$:
 $$
 \begin{tikzcd}[ampersand replacement=\&]
 \Aalg\ar[d, rightarrowtail]\ar[rr, "\upf" description] \&\&  \Balg\ar[d, rightarrowtail]\\[-4]
 \Aalg_{0}^1 \ar[rr,"\upf_{\,0}^{\,1}" description]\ar[d, phantom,"\vdots"] \&\& \Balg_{0}^1\ar[d, phantom,"\vdots"]\\[-4]
 \Aalg_{0}^{\upomega}\ar[rr, "\upf_{\,0}^{\,\upomega}" description]\ar[rd,equal] \&\&  \Balg_0^{\upomega}\ar[rd,equal] \\[-4]
 \& \Aalg_1^0\ar[d, rightarrowtail]\ar[rr, "\upf_{\,1}^{\,0}" description] \&\&  \Balg_1^0\ar[d, rightarrowtail]\\[-4]
 \& \Aalg_{1}^1 \ar[rr,"\upf_{\,1}^{\,1}" description]\ar[d, phantom,"\vdots"] \&\& \Balg_{1}^1\ar[d, phantom,"\vdots"]\\[-4]
 \& \Aalg_{1}^{\upomega}\ar[rr, "\upf_{\,1}^{\,\upomega}" description] \ar[rd, phantom, "\ddots"]\&\&  \Balg_1^{\upomega}\ar[rd, phantom, "\ddots"]\\[-4]
 \&\& {} \ar[d,phantom, "\vdots"] \&\& {} \ar[d, phantom, "\vdots"]\\[-4]
 \&\& {} \ar[rd, phantom, "\ddots"] \&\& {} \ar[rd, phantom, "\ddots"]\\[-4]
 \&\&\& \widetilde{\Aalg} \ar[rr,"\widetilde{\upf}" description] \&\& \widetilde{\Balg}
 \end{tikzcd}
 $$  
 and this fact is what we need to conclude that $\widetilde{\upf}$ in an equivalence in the extension model structure, since for this to hold we must check that evaluating $\widetilde{\upf}$ on $\upiota(\col\OpB)$ we get an equivalence. We use the usual inductive argument to show this. For inductive steps, $\upf^{\,\upn}_{\,\upalpha}$ being an equivalence on $\upiota(\col\OpB)$ implies that $\upf^{\,\upn+1}_{\,\upalpha}$ is so by an easy analysis of the pushouts appearing in \cite[Proposition 4.3.17]{white_bousfield_2018} (recall that we are working on $\Ch_{\mathbb{k}}$ with $\mathbb{k}$ field of characteristic zero). For transfinite steps, we use that quasi-isomorphisms are closed under transfinite composites.
\end{proof}

\begin{rem} The extreme restriction on the underlying model structure $\V$ in Proposition \ref{prop_ExtensionModelIsLeftProper} is chosen to maintain the technicalities at a minimum level. Various generalizations are possible, mainly due to results in \cite{carmona_enveloping_2024}, but we will not discuss them here. 
	
It is important to point out that a category of operadic algebras is not always left proper. The fact that this property holds when $\V=\Ch_{\mathbb{k}}$ ultimately relies on the extremely good behaviour of chain complexes over a field of characteristic zero.   	
\end{rem}

\end{section}

\begin{section}{Local-to-global cellularization}\label{sect_CausalityColocalization}
	This section is devoted to the presentation of a model category structure which encapsulates the homotopy theory of AQFTs that satisfy a natural local-to-global condition, and its interaction with the model structures defined in Section \ref{sect_TimeLocalization} that deal with the time-slice axiom. 

\begin{subsection}{The extension model structure on AQFTs}
First, we recall this canonical local-to-global condition (see \cite{benini_homotopy_2019, benini_operads_2021}).
\begin{defn}\label{defn_FredenhagenLocalToGlobalPrinciple}
Let $\Crect_{\diamond}^{\perp}\hookrightarrow \Crect^{\perp}$ be the inclusion of a full orthogonal subcategory and $$\upiota_{\sharp}\colon\Qft(\Crect_{\diamond}^{\perp})\rightleftarrows\Qft(\Crect^{\perp})\colon\upiota^*$$ the induced Quillen adjunction. Then, a field theory $\Aalg\in\QFT(\Crect^{\perp})$ is said to satisfy the $\Crect_{\diamond}$\emph{-local-to-global axiom} if the canonical map  
$
\mathbb{L}\upiota_{\sharp}\upiota^*\Aalg\to\Aalg
$ 
 is an equivalence, where $\mathbb{L}\upiota_{\sharp}$ denotes the derived functor of $\upiota_{\sharp}$. In the sequel,  $\QFT^{\Crect_{\diamond}}(\Crect^{\perp})$ denotes the full subcategory of $\QFT(\Crect^{\perp})$ spanned by field theories which satisfy the $\Crect_{\diamond}$\emph{-local-to-global axiom}.
\end{defn}
\begin{rem} In \cite{benini_operads_2021}, the authors proposed the above local-to-global principle as a substitute for Fredenhagen's universal construction \cite{fredenhagen_generalizations_1990}, which is proven to fail in certain situations because of the violation of Einstein causality, or in other words $\perp$-commutativity (\cite[Section 5]{benini_operads_2021}). 
\end{rem}

Letting aside the homotopical discussion for a moment, it is easy to see that the morphism of operads $\upiota\colon\Op_{\Crect_{\diamond}}^{\perp}\to\Op_{\Crect}^{\perp}$ induces an equivalence of categories
$$
\upiota_{\sharp}\colon\QFT(\Crect_{\diamond}^{\perp})\overset{\sim}{\longrightarrow}\left\lbrace\begin{matrix}
\Aalg\in\QFT(\Crect^{\perp})\text{ s.t.}\\ \upiota_{\sharp}\upiota^*\Aalg\xrightarrow{}\Aalg \text{ is an iso}
\end{matrix} \right\rbrace,
$$
which shows that the algebraic quantum field theories that satisfy the strict version of the $\Crect_{\diamond}$-local-to-global principle are completely determined by their restriction to spacetime regions in $\Crect_{\diamond}$. This fact clearly justifies the use of the cellularization discussed in Section  \ref{sect_ExtensionModelStructure}.

The operad inclusion $\upiota\colon\Op_{\Crect_{\diamond}}^{\perp}\hookrightarrow\Op_{\Crect}^{\perp}$ fulfills the requirements to apply Theorem \ref{thm_ExistenceOfExtensionModelAndProperties}. Hence, $\Qft(\Crect^{\perp})$ admits a cellularization, which we denote $\Qft^{\Crect_{\diamond}}(\Crect^{\perp})$, and  whose essential properties are collected in the following theorem.

\begin{thm}\label{thm_LocalToGlobalModel}
Let $\Crect_{\diamond}^{\perp}\hookrightarrow \Crect^{\perp}$ be the inclusion of a full orthogonal subcategory. Then, the model structure $\Qft^{\Crect_{\diamond}}(\Crect^{\perp})$ has:
\begin{itemize}
	\item as weak equivalences, morphisms of algebraic quantum field theories which are equivalences when evaluated on spacetime regions within $\Crect_{\diamond}$;
	\item as cofibrant objects, those cofibrant field theories in $\Qft(\Crect^{\perp})$ which satisfy the $\Crect_{\diamond}$-local-to-global axiom;
	\item as fibrations, the class of projective fibrations, i.e the fibrations in $\Qft(\Crect^{\perp})$.
\end{itemize}
Moreover, $\Qft^{\Crect_{\diamond}}(\Crect^{\perp})$ presents the homotopy theory of  $\QFT^{\Crect_{\diamond}}(\Crect^{\perp})$ and it is Quillen equivalent to $\Qft(\Crect_{\diamond}^{\perp})$.
\end{thm}
\begin{proof} All the statements are proven in Theorem \ref{thm_ExistenceOfExtensionModelAndProperties} except that $\Qft^{\Crect_{\diamond}}(\Crect^{\perp})$ presents the homotopy theory of $\QFT^{\Crect_{\diamond}}(\Crect^{\perp})$. This fact follows from the characterization of cofibrant objects in $\Qft^{\Crect_{\diamond}}(\Crect^{\perp})$ by the following chain of equivalences of categories
	$$
	\Ho\Qft^{\Crect_{\diamond}}(\Crect^{\perp})\simeq\Ho\big(\Qft^{\Crect_{\diamond}}(\Crect^{\perp})_{\textup{cofibrant}}\big)\simeq\Ho\QFT^{\Crect_{\diamond}}(\Crect^{\perp}).
	$$
\end{proof}

\begin{rem}\label{rem:LocaltoGlobalObjectsAreEasy} Not very surprisingly, the Quillen equivalence in Theorem \ref{thm_LocalToGlobalModel} between $ \Qft^{\Crect_{\diamond}}(\Crect^{\perp})$ and $\Qft(\Crect_{\diamond}^{\perp})$ implies that, in a homotopical sense, field theories satisfying the $\Crect_{\diamond}$-local-to-global axiom are completely characterized by their restriction to spacetime regions in $\Crect_{\diamond}$. 
\end{rem}

\end{subsection}

\begin{subsection}{Mixing localization and cellularization}
We discuss how the time-slice axiom and the local-to-global principle interact in terms of homotopy theory of AQFTs. We will need to repeat the distinction between the structural and the property-based approach in Section \ref{sect_TimeLocalization}. For these purposes, we start with some notation.

\begin{notat} We denote by $\QFT_{\h\Srect}^{\Crect_{\diamond}}(\Crect^{\perp})$ the full subcategory of $\QFT(\Crect^{\perp})$ spanned by homotopy $\Srect$-constant field theories that satisfy the $\Crect_{\diamond}$-local-to-global principle.
\end{notat}

\begin{paragraph}{Structural approach:} Our goal is to introduce the $\Crect_{\diamond}$-local-to-global principle into the model structure $\Qft_{\h\Srect}(\Crect^{\perp})$ by means of a cellularization as in the preceding section. In this situation, we work with $\mathfrak{L}_{\Srect}\Op_{\Crect}^{\perp}$ as the operad that governs the underlying algebraic structure. 

We want to make use again of the results of Section \ref{sect_ExtensionModelStructure}, so we must select a suitable map of operads $\OpB\to\OpN$. It is quite natural to consider the inclusion morphism  $\upiota_{\Srect}\colon (\mathfrak{L}_{\Srect}\Op_{\Crect}^{\perp})_{\diamond}\to\mathfrak{L}_{\Srect}\Op_{\Crect}^{\perp}$, where
$(\mathfrak{L}_{\Srect}\Op_{\Crect}^{\perp})_{\diamond}$ is the full suboperad of $\mathfrak{L}_{\Srect}\Op_{\Crect}^{\perp}$ spanned by the colors $\ob\Crect_{\diamond}$ (recall that $\mathfrak{L}_{\Srect}\Op_{\Crect}^{\perp}$ can be chosen to have colors $\ob\Crect$). Then, an application of Theorem \ref{thm_ExistenceOfExtensionModelAndProperties} yields:

\begin{thm}\label{thm_MixingModelWithStructure} Assume that $\Crect_{\diamond}^{\perp}\hookrightarrow\Crect^{\perp}$ is a full orthogonal subcategory and $\Srect$ a set of morphisms in $\Crect$. Then, there exists a Quillen model structure $\Qft_{\h\Srect}^{\Crect_{\diamond}}(\Crect^{\perp})$ on the category  $\QFT_{\h\Srect}(\Crect^{\perp})$ that satisfies:
	\begin{itemize}
		\item the weak equivalences are the maps that are equivalences when evaluated on spacetime regions within $\Crect_{\diamond}$;
	
		\item the cofibrant objects are the cofibrant objects in $\Qft_{\h\Srect}(\Crect^{\perp})$ which are  $(\mathfrak{L}_{\Srect}\Op_{\Crect}^{\perp})_{\diamond}$-colocal (Definition \ref{defn:ColocalAlgebra});
		
		\item the fibrations are the projective fibrations.
	\end{itemize}
\end{thm}
\end{paragraph}

\begin{paragraph}{Property-based approach:} Now, the underlying algebraic structure is parametrized by the operad $\Op_{\Crect}^{\perp}$, i.e. the underlying category will be $\QFT(\Crect^{\perp})\simeq \Op_{\Crect}^{\perp}\text{-}\Alg$ despite the previous choice  $\QFT_{\h\Srect}(\Crect^{\perp})\simeq \mathfrak{L}_{\Srect}\Op_{\Crect}^{\perp}\text{-}\Alg$. In order to present the homotopy theory of homotopy $\Srect$-constant field theories that satisfy the local-to-global principle, we need to perform two Bousfield localizations of the model structure $\Qft(\Crect^{\perp})$; each one deals with one axiom.
	
	We perform the cellularization described in Theorem \ref{thm_LocalToGlobalModel} to get  $\Qft^{\Crect_{\diamond}}(\Crect^{\perp})$. Due to Propositions \ref{prop_ExtensionModelCofibrantlyGenerated} and  \ref{prop_ExtensionModelIsLeftProper}, the model structure  $\Qft^{\Crect_{\diamond}}(\Crect^{\perp})$ admits the  localization at the set of maps $\mathcal{S}$ that appears in Proposition \ref{prop_WeakTimeSliceModel}. Let us denote this model structure by $\mathsf{L}_{\mathcal{S}}\Qft^{\Crect_{\diamond}}(\Crect^{\perp})$. 
	
	The fundamental result is the recognition of the bifibrant objects in this model category.
	
	\begin{thm}\label{thm_MixingModelWithProperties} The bifibrant objects in $\mathsf{L}_{\mathcal{S}}\Qft^{\Crect_{\diamond}}(\Crect^{\perp})$ are those bifibrant objects in $\Qft(\Crect^{\perp})$ which satisfy the homotopy $\Srect$-time-slice axiom and the $\Crect_{\diamond}$-local-to-global principle. 
	\end{thm}
	\begin{proof}
	The cellularization does not change the class of fibrations, and hence of fibrant objects, whereas the localization does not change the class of cofibrations. Thus, the result follows from Proposition \ref{prop_WeakTimeSliceModel} and Theorem \ref{thm_LocalToGlobalModel}.
	\end{proof}
	\begin{cor} The model structure  
	$\mathsf{L}_{\mathcal{S}}\Qft^{\Crect_{\diamond}}(\Crect^{\perp})$ presents the homotopy theory of algebraic field theories which satisfy the homotopy $\Srect$-time-slice axiom and the $\Crect_{\diamond}$-local-to-global principle, i.e.\;the homotopy theory of  $\QFT_{\h\Srect}^{\Crect_{\diamond}}(\Crect^{\perp})$.
	\end{cor}
    \begin{proof} This follows from the equivalences of categories
    	$$
    	\Ho \mathsf{L}_{\mathcal{S}}\Qft^{\Crect_{\diamond}}(\Crect^{\perp})\simeq 	\Ho\big( \mathsf{L}_{\mathcal{S}}\Qft^{\Crect_{\diamond}}(\Crect^{\perp})_{\textup{bifibrant}}\big)\simeq \Ho\QFT_{\h\Srect}^{\Crect_{\diamond}}(\Crect^{\perp}).
    	$$
    \end{proof}

	We finish by comparing the two approaches.
	
	\begin{thm}\label{thm_EquivalenceOfMixingModels} The localization morphism of operads $\ell\colon\Op_{\Crect}^{\perp}\to \mathfrak{L}_{\Srect}\Op_{\Crect}^{\perp}$ induces a Quillen equivalence of model categories
	$$
	\ell_{\sharp}\colon \mathsf{L}_{\mathcal{S}}\Qft^{\Crect_{\diamond}}(\Crect^{\perp})\rightleftarrows\Qft_{\h\Srect}^{\Crect_{\diamond}}(\Crect^{\perp})\colon\ell^*.
	$$
	\end{thm}
\begin{proof} By definition of $\Qft_{\h\Srect}(\Crect^{\perp})$, we already have a Quillen pair
	 $
	 \Qft(\Crect^{\perp})\rightleftarrows \Qft_{\h\Srect}(\Crect^{\perp})
	 $
	  induced by $\ell$. The idea is to prove that this pair descends to the cellularizations and the localization.

By Lemma \ref{lem_ColocalEquivalences}, the functor $\ell^*\colon \Qft_{\h\Srect}(\Crect^{\perp})\to \Qft(\Crect^{\perp})$ preserves the class of equivalences in Theorem \ref{thm_ExistenceOfExtensionModelAndProperties}. Hence, we have the corresponding Quillen pair between cellularizations,  $\Qft^{\Crect_{\diamond}}(\Crect^{\perp})\rightleftarrows\Qft_{\h\Srect}^{\Crect_{\diamond}}(\Crect^{\perp})$.
	
	Taking into account the universal property of the localization, it remains to check that the set $\mathcal{S}$ is sent to equivalences in $\Qft_{\h\Srect}^{\Crect_{\diamond}}(\Crect^{\perp})$. The same argument given in Theorem \ref{thm_ComparisonWeakTimeSliceModels} shows that this is the case. 
	
	Once constructed the Quillen pair of the statement, proving that it is a Quillen equivalence boils down to showing that it induces an equivalence of homotopy categories. This follows from the fact following facts:
	\begin{itemize}
		\item On the one hand,  $\Ho \mathsf{L}_{\mathcal{S}}\Qft^{\Crect_{\diamond}}(\Crect^{\perp})$ is equivalent to the homotopy category of its fibrant objects with equivalences between them being the class of maps
		$$
		\begin{Bmatrix}
			\upf\colon \Aalg\to \Balg \text{ such that  }\upf_{\Urect_{\diamond}}\colon\Aalg(\Urect_{\diamond})\overset{\sim}{\longrightarrow}\Balg(\Urect_{\diamond})\text{ for all }\Urect_{\diamond}\in\Crect_{\diamond}
		\end{Bmatrix}.
		$$ 
		To deduce such a claim just note that the class of $\mathcal{S}$-equivalences between fibrant objects in $\mathsf{L}_{\mathcal{S}}\Qft(\Crect^{\perp})$ is that of projective equivalences and that the equivalences in $\Qft^{\Crect_{\diamond}}(\Crect^{\perp})$ were identified in Lemma \ref{lem_ColocalEquivalences} as the previous class. Recall that the fibrant objects in $\Ho \mathsf{L}_{\mathcal{S}}\Qft^{\Crect_{\diamond}}(\Crect^{\perp})$ are the AQFTs that satisfy the homotopy $\Srect$-time slice axiom.
		
		\item 
		On the other hand, the homotopy category $\Ho\Qft_{\h\Srect}^{\Crect_{\diamond}}(\Crect^{\perp})$ is by construction the homotopy category of AQFTs satisfying the homotopy $\Srect$-time slice axiom seen as $\mathfrak{L}_{\Srect}\Op_{\Crect}^{\perp}$-algebras with equivalences being those maps of algebras that are equivalences when evaluated on every spacetime region $\Urect_{\diamond}\in\Crect_{\diamond}$.
	\end{itemize}
\end{proof}
\begin{cor} The model structure  $\Qft_{\h\Srect}^{\Crect_{\diamond}}(\Crect^{\perp})$ presents the homotopy theory of $\QFT_{\h\Srect}^{\Crect_{\diamond}}(\Crect^{\perp})$, i.e. algebraic field theories  which satisfy the homotopy $\Srect$-time-slice axiom and the $\Crect_{\diamond}$-local-to-global principle.
\end{cor}

\paragraph{AQFTs over $\Crect_{\diamond}$ satisfying homotopy time-slice axiom.} As commented in Remark \ref{rem:LocaltoGlobalObjectsAreEasy}, the homotopy theory associated to $\Qft^{\Crect_{\diamond}}(\Crect^{\perp})$ is easy in the sense that is coincides with the homotopy theory associated to $\Qft(\Crect_{\diamond}^{\perp})$. One natural question is if this kind of result extends to the homotopy theory associated to $\Qft^{\Crect_{\diamond}}_{\h\Srect}(\Crect^{\perp})$, i.e.\ if it is equivalent to the homotopy theory of AQFTs over $\Crect_{\diamond}^{\perp}$ satisfying some conditions (as time-slice like axioms). We now address this question, giving a general result under a quite restrictive hypothesis, we point towards technicalities arising for the general case and we also discuss two additional examples.

First, we observe that in the \emph{structural approach} of this subsection a choice was made: we applied our results in Section \ref{sect_ExtensionModelStructure} to the inclusion morphism $\upiota_{\Srect}\colon (\mathfrak{L}_{\Srect}\Op_{\Crect}^{\perp})_{\diamond}\to \mathfrak{L}_{\Srect}\Op_{\Crect}^{\perp}$. Revisiting that discussion, one can see that the operad $(\mathfrak{L}_{\Srect}\Op_{\Crect}^{\perp})_{\diamond}$ is too abstract, in the sense that we do not have any control of how it looks like, and that obstructed us to connect the original  $\Crect_{\diamond}$-local-to-global property with being $(\mathfrak{L}_{\Srect}\Op_{\Crect}^{\perp})_{\diamond}$-colocal.

If one instead picks the morphism of operads 
$
\mathfrak{j}\colon \mathfrak{L}_{\Srect_{\diamond}}\Op_{\Crect_{\diamond}}^{\perp}\rightarrow \mathfrak{L}_{\Srect}\Op_{\Crect}^{\perp},
$
where $\Srect_{\diamond}$ denotes those maps in $\Srect$ between objects in $\Crect_{\diamond}$, it is possible to relate $\mathfrak{L}_{\Srect_{\diamond}}\Op_{\Crect_{\diamond}}^{\perp}$-colocality with the $\Crect_{\diamond}$-local-to-global property. Before proving this claim, note that $\mathfrak{j}$ fits into the following commutative square of operads
\begin{equation}\label{eqt_DiagramOfLocalizationsOfOperads}
	\begin{tikzcd}[ampersand replacement=\&]
		\Op_{\Crect}^{\perp}\ar[r,"\ell"] \& \mathfrak{L}_{\Srect}\Op_{\Crect}^{\perp}\\
		\Op_{\Crect_{\diamond}}^{\perp}\ar[r, "\ell_{\diamond}"'] \ar[u, "\upiota"]\ar[ru, "\upphi" description] \& \mathfrak{L}_{\Srect_{\diamond}}\Op_{\Crect_{\diamond}}^{\perp},\ar[u, "\mathfrak{j}"']
	\end{tikzcd}.
\end{equation}
\begin{prop}\label{prop_LocaltoGlobalPrincipleInMixedModelWithStructures}
	The following conditions are equivalent for $\Aalg\in\QFT_{\h\Srect}(\Crect^{\perp})$:
	\begin{itemize}
		\item
		($\mathfrak{L}_{\Srect_{\diamond}}\Op_{\Crect_{\diamond}}^{\perp}$-colocality) $\mathbb{L}\mathfrak{j}_{\sharp}\mathfrak{j}^*\Aalg\to\Aalg$ is an equivalence.
		\item (modified $\Crect_{\diamond}$-local-to-global property) $\mathbb{L}\upphi_{\sharp}\upphi^*\Aalg\to\Aalg$ is an equivalence.
	\end{itemize}
\end{prop}	
\begin{proof} The commutative square (\ref{eqt_DiagramOfLocalizationsOfOperads}) induces the commutative diagram of adjoint functors
	$$
	\begin{tikzcd}[ampersand replacement=\&]
		\Ho\Qft(\Crect^{\perp}) \ar[r, shift left=.75ex]\ar[r,leftarrow,shift right=.75ex] \& \Ho\Qft_{\h\Srect}(\Crect^{\perp})\\ 
		\Ho\Qft(\Crect_{\diamond}^{\perp}) \ar[r, shift left=.75ex]\ar[r,leftarrow,shift right=.75ex]\ar[u, shift left=.75ex]\ar[u,leftarrow,shift right=.75ex] \& \Ho\Qft_{\h\Srect_{\diamond}}(\Crect^{\perp}_{\diamond}).\ar[u, shift left=.75ex]\ar[u,leftarrow,shift right=.75ex]
	\end{tikzcd}
	$$
	Therefore, with appropriate compositions of adjoints, we obtain the diagram in $\Ho\QFT_{\h\Srect}(\Crect^{\perp})$
	$$
	\begin{tikzcd}[ampersand replacement=\&]
		\& \mathbb{L}\upphi_{\sharp}\upphi^*\Aalg\ar[rrd, "(iii)"] \& \& \\
		\mathbb{L}\mathfrak{j}_{\sharp}\cdot\mathbb{L}\ell_{\diamond\sharp}\ell_{\diamond}^*\cdot\mathfrak{j}^*\Aalg \ar[ru, "\overset{(i)}{\simeq}"]\ar[rd, "\underset{(ii)}{\simeq}"']\& \&  \& \Aalg.\\
		\& \mathbb{L}\mathfrak{j}_{\sharp}\mathfrak{j}^*\Aalg\ar[rru, "(iv)"']
	\end{tikzcd}
	$$
	On the one hand, $(i)$ is an equivalence (iso in the homotopy category) by the commutativity of the above square. On the other hand, $(ii)$ is an equivalence because, by Theorem \ref{thm_ComparisonWeakTimeSliceModels}, the adjunction $\mathbb{L}\ell_{\diamond\sharp}\colon \Ho\QFT(\Crect_{\diamond}^{\perp})\rightleftarrows\Ho\QFT_{\h\Srect}(\Crect_{\diamond}^{\perp})\colon\ell_{\diamond}^*$ is a reflection. Thus, $(iii)$ is an equivalence iff $(iv)$ is so.
\end{proof}
\begin{rem} The modified $\Crect_{\diamond}$-local-to-global principle in Proposition \ref{prop_LocaltoGlobalPrincipleInMixedModelWithStructures} means that AQFTs satisfying this principle can be reconstructed from its restriction to $\Crect_{\diamond}$ in a strong sense.
\end{rem}

However, a technical assumption is required on $\Srect$ in order to make Theorem \ref{thm_ExistenceOfExtensionModelAndProperties} work for $\mathfrak{j}$. Instead of finding general conditions for these purposes, we consider a quite restrictive hypothesis which ensures that everything matches in the best possible way (see the trapezoid in the lower right corner of the diagram in Glossary  \ref{sect_ComparisonFieldTheories}). For more information, see Remark  \ref{rem_ComplicationsWhenMixingModelsWithStructure}.

\begin{hyp}\label{hyp_ExtremelyGoodLocalizingSet}
	The set of maps $\Srect$ belongs completely to the full subcategory $\Crect_{\diamond}$, i.e.\ $\Srect_{\diamond}=\Srect$.
\end{hyp}

\begin{thm}\label{thm_ComparingMixingWithDIAMONDS} Assume that $\Crect_{\diamond}^{\perp}\hookrightarrow\Crect^{\perp}$ is a full orthogonal subcategory and $\Srect$ a set of morphisms in $\Crect$ satisfying Hypothesis \ref{hyp_ExtremelyGoodLocalizingSet}. Then,
	\begin{itemize}
		\item $\Aalg\in \Qft_{\h\Srect}^{\Crect_{\diamond}}(\Crect^{\perp})$ is cofibrant iff it is cofibrant in $\Qft_{\h\Srect}(\Crect^{\perp})$ and it satisfies the modified $\Crect_{\diamond}$-local-to-global property in Proposition \ref{prop_LocaltoGlobalPrincipleInMixedModelWithStructures}.
		
		\item The Quillen adjunction $\Qft_{\h\Srect_{\diamond}}(\Crect^{\perp}_{\diamond})\rightleftarrows \Qft_{\h\Srect}^{\Crect_{\diamond}}(\Crect^{\perp})$ is a Quillen equivalence.
	\end{itemize}
\end{thm}

\begin{proof}
	Hypothesis \ref{hyp_ExtremelyGoodLocalizingSet} implies that the canonical map   $\mathfrak{L}_{\Srect_{\diamond}}\Op_{\Crect_{\diamond}}^{\perp}\to(\mathfrak{L}_{\Srect}\Op_{\Crect}^{\perp})_{\diamond}$ is an equivalence because $(\mathfrak{L}_{\Srect}\Op_{\Crect}^{\perp})_{\diamond}$ satisfies the universal property that characterizes $\mathfrak{L}_{\Srect_{\diamond}}\Op_{\Crect_{\diamond}}^{\perp}$. Therefore, the first point follows directly from Proposition \ref{prop_LocaltoGlobalPrincipleInMixedModelWithStructures}. The second point follows from the chain of equivalences
\begin{align*}
\Ho\Qft_{\h\Srect}^{\Crect_{\diamond}}(\Crect^{\perp})
	& \overset{(i)}{\simeq} \Ho\mathsf{L}_{\mathcal{S}}\Qft^{\Crect_{\diamond}}(\Crect^{\perp})\\
	& \overset{(ii)}{\simeq} \Ho\mathsf{L}_{\mathcal{S}_{\diamond}}\Qft^{\Crect_{\diamond}}(\Crect^{\perp})\\
	& \overset{(iii)}{\simeq} \Ho\mathsf{L}_{\mathcal{S}_{\diamond}}\Qft(\Crect_{\diamond}^{\perp})\\
	& \overset{(iv)}{\simeq} \Ho\Qft_{\h\Srect_{\diamond}}(\Crect_{\diamond}^{\perp}).
\end{align*}
$(i)$ is a particular instance of Theorem \ref{thm_ComparisonWeakTimeSliceModels}; $(ii)$ is ensured by Hypothesis \ref{hyp_ExtremelyGoodLocalizingSet};  $(iii)$ is immediate since we are localizing Quillen equivalent model structures at the same set of maps; $(iv)$ same as $(i)$.
\end{proof}

\begin{rem}\label{rem_ComplicationsWhenMixingModelsWithStructure}
	It is not true that $\mathfrak{j}\colon\mathfrak{L}_{\Srect_{\diamond}}\Op_{\Crect_{\diamond}}^{\perp}\to \mathfrak{L}_{\Srect}\Op_{\Crect}^{\perp}$ satisfies Hypothesis \ref{hyp_OperadicLanHomotopicallyFF} in general, so Theorem \ref{thm_ExistenceOfExtensionModelAndProperties} could not be applicable. Conditions such as 
	$
	\mathfrak{L}_{\Srect_{\diamond}}\Crect_{\diamond}(\Urect,\Vrect)\rightarrow\mathfrak{L}_{\Srect}\Crect(\Urect,\Vrect)
	$
	being an equivalence $\forall \Urect,\,\Vrect\in\Crect_{\diamond}$ (homotopical fully-faithfulness) might be  sufficient for these purposes, although some more work will be needed. Furthermore, even a manageable condition on $\Srect$ that ensures this  seems to be hard to find, so we prefer to avoid this problematic in here. Nevertheless, we note that in practice one may circumvent this problematic as discussed below. 
\end{rem}	

To illustrate Theorem \ref{thm_ComparingMixingWithDIAMONDS}, take $\Crect^{\perp}$ to be the orthogonal category $\mathsf{Loc}_{m}^{\perp}$ of  oriented, time-oriented, globally hyperbolic Lorentzian $ m$-manifolds with morphisms isometric open embeddings preserving (time-)orientation and whose image is causally convex. That is, $\QFT(\mathsf{Loc}_{m}^{\perp})$ is the category of locally covariant $m$-dimensional quantum field theories. Consider the full subcategory  $\mathsf{Loc}_{m,\diamond}\hookrightarrow\mathsf{Loc}_{m}$ spanned by spacetimes diffeomorphic to $\mathbb{R}^m$ and $\Srect_{\diamond}$ to be the set of Cauchy-morphisms among spacetimes in $\mathsf{Loc}_{m,\diamond}$, i.e.\ embeddings  $\upf\colon \Mrect\hookrightarrow \mathsf{N}$ in $\mathsf{Loc}_{m,\diamond}$ such that $\upf(\Mrect)$ contains a Cauchy surface for $\mathsf{N}$. Then, 
\begin{itemize}
	\item the homotopy $\Srect_{\diamond}$-time-slice axiom asserts that a locally covariant AQFT sends Cauchy-morphisms between objects in $\mathsf{Loc}_{m,\diamond}$ to quasi-isomorphisms (note that nothing is imposed over more general Cauchy-morphisms);
	
	\item the $\mathsf{Loc}_{m,\diamond}$-local-to-global principle asserts that a locally covariant AQFT can be reconstructed from its restriction to spacetimes diffeomorphic to $\mathbb{R}^m$. 
\end{itemize}

 Then, Theorem \ref{thm_ComparingMixingWithDIAMONDS} yields an equivalence of homotopy categories
$$
\Ho\begin{Bmatrix}
\text{locally covariant }m\text{-dim AQFTs}\\
	\text{satisf.\  }\mathsf{Loc}_{m,\diamond}\textup{-local-to-global princ.}\\
	\text{and homotopy }\Srect_{\diamond}\text{-time-slice axiom}
\end{Bmatrix} \simeq
\Ho\begin{Bmatrix}
\text{locally covariant }m\text{-dim AQFTs}\\
\text{def.\ on spacetimes }\Mrect
\text{ diffeo to }\mathbb{R}^m\\
\text{satisf.\ homotopy   }\Srect_{\diamond}\text{-time-slice axiom}
\end{Bmatrix}.
$$
It may seem immediate since Theorem \ref{thm_LocalToGlobalModel} already gives a Quillen equivalence 
$$
\Qft^{\mathsf{Loc}_{m,\diamond}}(\mathsf{Loc}_{m}^{\perp})\simeq \Qft(\mathsf{Loc}_{m,\diamond}^{\perp}),
$$
which yields the equivalence of homotopy categories not taking into account the homotopy time-slice axiom. However, if we enlarge $\Srect_{\diamond}$ and take $\Srect$ to be the set of Cauchy-morphisms between any pair of spacetimes, the analogous equivalence of homotopy categories is not known to hold. Let us present a baby example to show how this may fail in general.

\begin{ex}\label{ex_ComparisonOfMixingOnDiamondsNOTHOLDINGENERAL} We are going to define an orthogonal category $\Crect^{\perp}$, a full orthogonal subcategory  $\Crect_{\diamond}^{\perp}\hookrightarrow\Crect^{\perp}$ and a set of morphisms $\Srect$ in $\Crect$ (not contained in $\Crect_{\diamond}$) for which 	
	$
	\Qft_{\h\Srect}^{\Crect_{\diamond}}(\Crect^{\perp})
	$
	and 
	$ 
	\Qft_{\h\Srect_{\diamond}}(\Crect_{\diamond}^{\perp})
	$
are not Quillen equivalent, where $\Srect_{\diamond}$ is the subset of morphisms of  $\Srect$ within $\Crect_{\diamond}$.

A simple choice for this purpose is: let $\Crect$ be the category $\left\{\diamond_1\leftarrow \Xrect\to \diamond_2\right\}$ (the orthogonality relation must be trivial since there are no maps with the same target),  $\Crect_{\diamond}=\left\{\diamond_1,\diamond_2\right\}$ and $\Srect$ be the whole set of arrows of $\Crect$. By definition $\Srect_{\diamond}$ is empty. Then, it is easy to check 
$$
\mathsf{ev}_{\diamond_1}\colon\Ho\Qft_{\h\Srect}^{\Crect_{\diamond}}(\Crect^{\perp})\overset{\sim}{\rightarrow} \Ho\mathsf{dga}_{\mathbb{k}} \quad  \text{and} \quad (\mathsf{ev}_{\diamond_1},\mathsf{ev}_{\diamond_2})\colon\Ho\Qft_{\h\Srect_{\diamond}}(\Crect^{\perp}_{\diamond})\overset{\sim}{\rightarrow} \Ho\mathsf{dga}_{\mathbb{k}}^{\times 2},
$$
where $\Ho\mathsf{dga}_{\mathbb{k}}$ is the homotopy category of dg-algebras over $\mathbb{k}$, and so, $
\Qft_{\h\Srect}^{\Crect_{\diamond}}(\Crect^{\perp})$ cannot be Quillen equivalent to  
$\Qft_{\h\Srect_{\diamond}}(\Crect_{\diamond}^{\perp})$.
\end{ex}

Let us close this section by discussing examples not covered by Theorem \ref{thm_ComparingMixingWithDIAMONDS}.

\begin{ex}[AQFTs on a fixed spacetime]\label{ex_COpenMixingDiamonds} Let $\Mrect$ be a fixed spacetime (oriented and time-oriented globally hyperbolic Lorentzian manifold) with or without timelike boundary and let  $\mathsf{COpen}(\Mrect)^{\perp}$ be the orthogonal category of causally convex open subsets of $\Mrect$ with the causal orthogonality relation. In this situation, $\QFT(\Mrect)\equiv\QFT(\mathsf{COpen}(\Mrect)^{\perp})$ is a category of AQFTs defined on $\Mrect$ in the spirit of Haag-Kastler's axiomatization (see \cite{benini_algebraic_2018,haag_algebraic_1964}). 
	
Consider $\mathsf{COpen}_{\diamond}(\Mrect)\hookrightarrow \mathsf{COpen}(\Mrect)$	to be the full subcategory spanned by causal diamonds in $\Mrect$ and $\Srect$ to be the set of inclusions of causally convex open subsets of $\Mrect$ sharing a Cauchy surface. To alleviate the notation, let us replace the superindex $\mathsf{COpen}_{\diamond}(\Mrect)$ by $\diamond$ to denote the associated local-to-global principle, e.g.\ $\Qft^{\diamond}(\Mrect)$ instead of $\Qft^{\mathsf{COpen}_{\diamond}(\Mrect)}(\mathsf{COpen}(\Mrect)^{\perp})$.

\begin{prop}\label{prop_HaagKastlerMixingDIAMONDS}
	The inclusion of orthogonal categories $\upiota\colon\mathsf{COpen}_{\diamond}(\Mrect)\hookrightarrow \mathsf{COpen}(\Mrect)$ induces a Quillen equivalence
	$
	\Qft_{\h\Srect_{\diamond}}(\mathsf{COpen}_{\diamond}(\Mrect)^{\perp}) \rightleftarrows \Qft_{\h\Srect}^{\diamond}(\Mrect). 
	$
\end{prop}
\begin{proof}
By Theorems \ref{thm_ComparisonWeakTimeSliceModels} and \ref{thm_EquivalenceOfMixingModels}, we can equivalently prove that the Quillen equivalence 
$
\upiota_{\sharp}\colon\Qft(\mathsf{COpen}_{\diamond}(\Mrect)^{\perp}) \rightleftarrows \Qft^{\diamond}(\Mrect)\colon\!\upiota^*
$ given in Theorem \ref{thm_LocalToGlobalModel} descends to a Quillen equivalence 
$$
\upiota_{\sharp}\colon\Lrect_{\mathcal{S}_{\diamond}}\Qft(\mathsf{COpen}_{\diamond}(\Mrect)^{\perp}) \rightleftarrows \Lrect_{\mathcal{S}}\Qft^{\diamond}(\Mrect)\colon\!\upiota^*.
$$ 

Let us first check that it descends to a Quillen pair.

 Since the left adjoint preserves cofibrations, we are reduced to show that the right adjoint preserves fibrations between fibrant objects by \cite[Proposition 8.5.4]{hirschhorn_model_2003}. Using \cite[Proposition 4.30]{barwick_left_2010},  Proposition \ref{prop_CharacterizationOfTIMESLICELOCALObjects} and Theorem \ref{thm_MixingModelWithProperties}, we know that fibrations between fibrant objects in both Bousfield localizations are the projective fibrations between projectively fibrant AQFTs which satisfy the corresponding homotopy time-slice axiom. Observe that the right adjoint is just the restriction along $\mathsf{COpen}_{\diamond}(\Mrect)\hookrightarrow \mathsf{COpen}(\Mrect)$ to conclude that we have a Quillen pair.
 
 Finally, let us prove that the Quillen pair is a Quillen equivalence by checking that the derived unit and counit are equivalences in the respective model structure.
 \begin{itemize}
 	\item derived counit: we must show that $\mathbb{L}\upiota_{\sharp}\upiota^*\Balg\to\Balg$ is a proj-equivalence for any bifibrant object $\Balg\in\Lrect_{\mathcal{S}}\Qft^{\diamond}(\Mrect)$. Theorem  \ref{thm_MixingModelWithProperties} concludes the claim by definition of $\Crect_{\diamond}$-local-to-global principle. 
 
 	\item derived unit: we must show that $\Aalg\to\mathbb{R}\upiota^*\upiota_{\sharp}\Aalg$ is a proj-equivalence for any bifibrant object $\Aalg\in\Lrect_{\mathcal{S}_{\diamond}}\Qft(\mathsf{COpen}_{\diamond}(\Mrect)^{\perp})$. The subtle point here is that the derived functor $\mathbb{R}\upiota^*$ is computed by applying $\upiota^*$ to a fibrant replacement of $\upiota_{\sharp}\Aalg$ in $\Lrect_{\mathcal{S}}\Qft^{\diamond}(\Mrect)$, which a priori may modify the value of $\upiota_{\sharp}\Aalg$ on causal diamonds. Let us see that this is not the case by constructing the fibrant replacement explicitly in this case. 
 	
 	Recall from \cite[Section 2]{benini_algebraic_2018} that taking Cauchy-development of spacetime regions yields an orthogonal functor $\Drect\colon \mathsf{COpen}(\Mrect)\to \mathsf{COpen}(\Mrect)$ equipped with a natural transformation from the identity functor $\id\Rightarrow \Drect$. On any region, this transformation is simply the canonical inclusion $\Urect\hookrightarrow\Drect(\Urect)$  of the region within its Cauchy-development and so it is a Cauchy-morphism. This data induces an endofunctor $\Drect^*\colon \QFT(\Mrect)\to \QFT(\Mrect)$ together with a natural transformation $\id\Rightarrow \Drect^*$. Choose a proj-fibrant replacement functor $\mathsf{R}$ for $\Qft(\Mrect)$ and consider the composition $\mathfrak{D}=\Drect^*\mathsf{R}$. We claim that $\mathfrak{D}$ serves as a fibrant replacement for $\Lrect_{\mathcal{S}}\Qft^{\diamond}(\Mrect)$. It suffices to check that for any $\Balg\in \QFT(\Mrect)$ and any map $\Balg\to \widehat{\Balg}$ in $\Ho\Qft(\Mrect)$ such that $\widehat{\Balg}$ satisfies the homotopy $\Srect$-time-slice axiom, there is a unique factorization through $\Balg\to \mathfrak{D}\Balg$ since $\mathfrak{D}\Balg$ automatically satisfies the homotopy $\Srect$-time-slice axiom and is projectively fibrant. To prove that, consider the commutative square in $\Ho\Qft(\Mrect)$
 	$$
 	\begin{tikzcd}[ampersand replacement=\&]
 		\Balg\ar[r]\ar[d] \& \widehat{\Balg}\ar[d,"\cong"]\\
 		\mathfrak{D}\Balg\ar[ru, dashed, "!"] \ar[r] \& \mathfrak{D}\widehat{\Balg}
 	\end{tikzcd}\quad .
 	$$ 
 	The vertical map on the right hand side is an isomorphism in $\Ho\Qft(\Mrect)$ since the natural transformation  $\id\Rightarrow \Drect$ is a Cauchy-morphism when evaluated on any spacetime region. Note that $\id\Rightarrow\Drect$ restricts to causal diamonds, since $\Drect(\Vrect)$ is globally hyperbolic with Cauchy surface equal to that of $\Vrect$.
 	
 	Using $\mathfrak{D}$, we can compute up to proj-equivalence
 	$
    \mathbb{R}\upiota^*\upiota_{\sharp}\Aalg\simeq \upiota^*\mathfrak{D}\upiota_{\sharp}\Aalg\simeq \upiota^*\Drect^*\upiota_{\sharp}\Aalg.
 	$
   Hence, evaluating the derived unit for $\Aalg$ on any $\Vrect\in\mathsf{COpen}_{\diamond}(\Mrect)$ we obtain
 	$
 	\Aalg(\Vrect)\rightarrow \upiota_{\sharp}\Aalg(\Drect(\Vrect))
 	$
 	which is a proj-equivalence because: (i) $\Drect(\Vrect)$ belongs to $\mathsf{COpen}_{\diamond}(\Mrect)$ since $\Drect$ restricts to causal diamonds; (ii)  $\Aalg$ satisfies the homotopy $\Srect_{\diamond}$-time-slice axiom by hypothesis (note that  $\upiota_{\sharp}\Aalg(\Drect(\Vrect))\cong \Aalg(\Drect(\Vrect))$ by (i)).
 \end{itemize} 
\end{proof}

\end{ex}

\begin{ex}[Locally covariant conformal 2D AQFTs ]\label{ex_CLocMixingDiamonds}
	 Let $\mathsf{CLoc}_2$ be the category of spacetimes (oriented and time-oriented globally hyperbolic Lorentzian manifolds) with morphisms embeddings preserving (time-)orientation and conformal structure and whose image is open and causally convex. We equip $\mathsf{CLoc}_2$ with the causal orthogonality relation. In this situation, $\QFT(\mathsf{CLoc}_2^{\perp})$ is the category of locally covariant conformal AQFTs in two spacetime dimensions. Additionally, consider $\mathsf{CLoc}_{2,\diamond}\hookrightarrow \mathsf{CLoc}_2$ to be the full subcategory spanned by spacetimes $\mathsf{N}$ diffeomorphic to $\mathbb{R}^2$ and $\Srect$ to be the set of Cauchy-morphisms, i.e.\ embeddings $\upf\colon \Mrect\hookrightarrow \mathsf{N}$ in $\mathsf{CLoc}_{2}$ such that $\upf(\Mrect)$ contains a Cauchy surface for $\mathsf{N}$. As in Example \ref{ex_COpenMixingDiamonds}, let us replace the superindex $\mathsf{CLoc}_{2,\diamond}$ by $\diamond$ to denote the associated local-to-global principle.
	 
	 \begin{prop}\label{prop_CLoc2MixingDIAMONDS}
	 	 	The inclusion of orthogonal categories $\mathsf{CLoc}_{2,\diamond}\hookrightarrow \mathsf{CLoc}_2$ induces a Quillen equivalence
	 	$
	 	\Qft_{\h\Srect_{\diamond}}(\mathsf{CLoc}_{2,\diamond}^{\perp}) \rightleftarrows \Qft_{\h\Srect}^{\diamond}(\mathsf{CLoc}_2^{\perp}). 
	 	$
	 \end{prop}
	 \begin{proof} Analogous to the proof of Proposition \ref{prop_HaagKastlerMixingDIAMONDS} using that (the subcategory of connected manifolds in) $\mathsf{CLoc}_{2,\diamond}^{\perp}$ is equivalent to the orthogonal category described in \cite[Corollary 3.5]{benini_skeletal_2022}.
	 \end{proof}
	 
\end{ex}

\end{paragraph}

\end{subsection}

\end{section}

\begin{section}{Glossary: Comparison of homotopy theories}\label{sect_ComparisonFieldTheories}
	This glossary contains a table summarizing all the model structures discussed in this work and a diagram that represents the relations between them.\vspace*{-1.5mm}

\begin{table}[H]
	\resizebox{\textwidth}{!}{%
		$
		\begin{array}{|l|l|cc|c|l|}
		\hline 
		\multicolumn{2}{|c|}{\textit{Homotopy theory of}} & \multicolumn{2}{c|}{\textit{Quillen model structure}} & \multicolumn{2}{c|}{\textit{Underlying algebraic structure}} \\ \hline \hline
		\multicolumn{2}{|c|}{
		\vphantom{\begin{array}{c}
			\text{\small{A}}\\[-.5mm]
			\text{\small{A}}
			\end{array}}
			\text{\small{algebraic field theories}}
		}                                  &                  \Qft(\Crect^{\perp})           &    \text{\small{(Theorem \ref{thm_ProjectiveModelOnPlainAQFTs})}}                       & \multicolumn{2}{c|}{\text{algebras over  } \Op_{\Crect}^{\perp}}                                 \\ \hline
		\multicolumn{2}{|c|}{
			\vphantom{\begin{array}{c}
				\text{\small{A}}\\[-.5mm]
				\text{\small{A}}
				\end{array}}
			\text{\small{strictly S-constant field theories}}
		}                                     &                 \Qft_{\Srect}(\Crect^{\perp})                   &     \text{\small{(Notation \ref{notat_ProjModelonAQFTsWithStrictTimeSliceAxiom})}}                      & \multicolumn{2}{c|}{\text{algebras over }\Op^{\perp_{\Srect}}_{\Crect_{\Srect}}}                                  \\ \hline
		\multicolumn{2}{|c|}{
			\vphantom{\begin{array}{c}
				\text{\small{A}}\\[-.5mm]
				\text{\small{A}}
				\end{array}}
			\text{\small{homotopy S-constant  field theories}}
		}                          &               \Qft_{\h\Srect}(\Crect^{\perp})           &     \text{\small{(Proposition \ref{prop_STRUCTUREDweakTimeSliceAQFTModel})}}                      & \multicolumn{2}{c|}{\text{algebras over }\mathfrak{L}_{\Srect}\Op_{\Crect}^{\perp}}                                  \\ \hline
		\multicolumn{2}{|c|}{
			\vphantom{\begin{array}{c}
				\text{\small{A}}\\[-.5mm]
				\text{\small{A}}
				\end{array}}
			\text{\small{homotopy S-constant field theories}}
		}                                   &                 \mathsf{L}_{\mathcal{S}}\Qft(\Crect^{\perp})           &    \text{\small{(Proposition \ref{prop_WeakTimeSliceModel})}}                       & \multicolumn{2}{c|}{\text{algebras over }\Op^{\perp}_{\Crect}}                                  \\ \hline
		\multicolumn{2}{|c|}{
			\parbox[c][15mm]{0mm}{}
			\begin{array}{c}
				\text{\small{algebraic field theories}}\\[-.5mm]
				\text{\small{satisfying local-to-global principle}}
				\end{array}
		}                                  &            \Qft^{\Crect_{\diamond}}(\Crect^{\perp})                &    \text{\small{(Theorem \ref{thm_LocalToGlobalModel})}}                   & \multicolumn{2}{c|}{\text{algebras over }\Op_{\Crect}^{\perp}}                                  \\ \hline
		\multicolumn{2}{|c|}{
			\parbox[c][15mm]{0mm}{}
			\begin{array}{c}
			\text{\small{homotopy S-constant field theories}}\\[-.5mm]
			\text{\small{satisfying local-to-global principle}}
			\end{array}
		}    
	
	                          &       \Qft_{\h\Srect}^{\Crect_{\diamond}}(\Crect^{\perp})                     &             \text{\small{(Theorem \ref{thm_MixingModelWithStructure})}}              & \multicolumn{2}{c|}{\text{algebras over }\mathfrak{L}_{\Srect}\Op_{\Crect}^{\perp}}                                  \\ \hline
		\multicolumn{2}{|c|}{
			\parbox[c][15mm]{0mm}{}
			\begin{array}{c}
			\text{\small{homotopy S-constant field theories}}\\[-.5mm]
			\text{\small{satisfying local-to-global principle}}
			\end{array}
		}                                     &             \mathsf{L}_{\mathcal{S}}\Qft^{\Crect_{\diamond}}(\Crect^{\perp})               &            \text{\small{(Theorem \ref{thm_MixingModelWithProperties})}}               & \multicolumn{2}{c|}{\text{algebras over }\Op_{\Crect}^{\perp}}
		     \\ \hline
		\end{array}
		$
	}
\end{table} \vspace*{-2.5mm}

In the following diagram are displayed these Quillen model categories with Quillen adjunctions between them, oriented by their left adjoints.
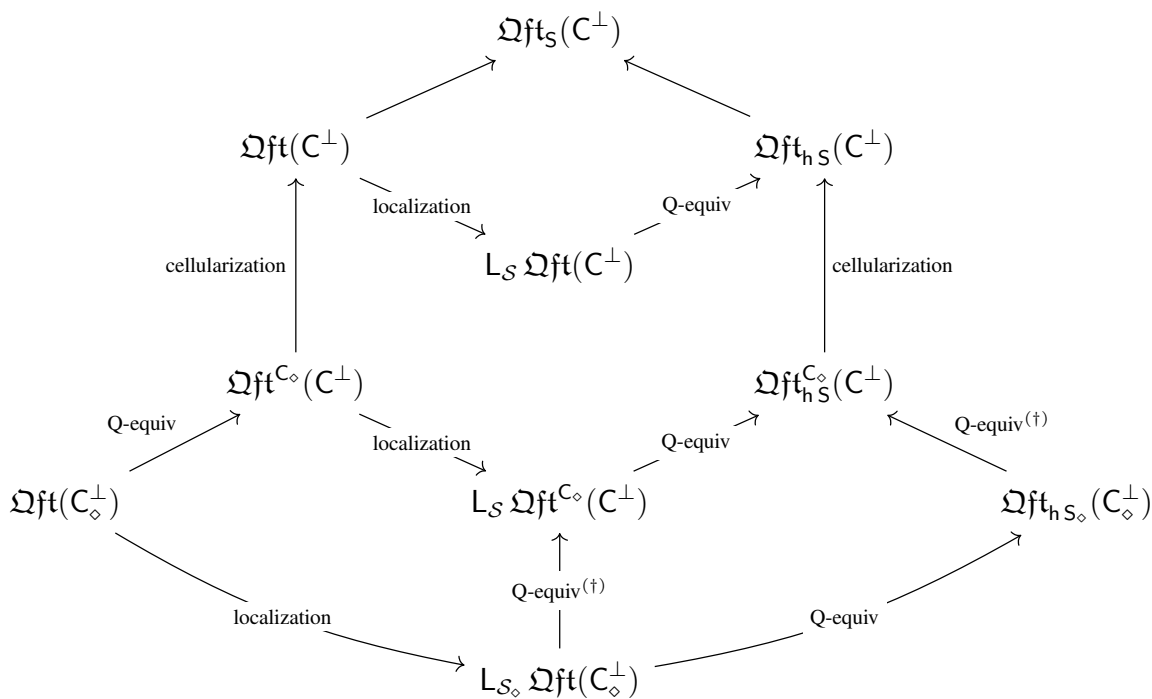
\begin{equation}\tag{$\mathscr{C}\hspace*{-0.7mm}\textit{omparison}$}\label{eqt_ComparisonAQFTs}
\begin{tikzcd}[ampersand replacement=\&]
\& \& \Qft_{\Srect}(\Crect^{\perp})\ar[rd, leftarrow] \&  \& \\
\& \Qft(\Crect^{\perp})\ar[ru]\ar[rd,"\text{localization}" description]\ar[dd,"\text{cellularization}"', leftarrow] \&  \& \Qft_{\h\Srect}(\Crect^{\perp}) \ar[dd,"\text{cellularization}", leftarrow]\& \\
\& \& \mathsf{L}_{\mathcal{S}}\Qft(\Crect^{\perp}) \ar[ru, "\text{Q-equiv}" description]\& \& \\
\& \Qft^{\Crect_{\diamond}}(\Crect^{\perp})\ar[rd,"\text{localization}" description] \&  \& \Qft^{\Crect_{\diamond}}_{\h\Srect}(\Crect^{\perp}) \& \\
\Qft(\Crect^{\perp}_{\diamond})\ar[rrdd,"\text{localization}" description, bend right=10]\ar[ru, "\text{Q-equiv}"]\& \& \mathsf{L}_{\mathcal{S}}\Qft^{\Crect_{\diamond}}(\Crect^{\perp}) \ar[ru, "\text{Q-equiv}" description] \& \&\Qft_{\h\Srect_{\diamond}}(\Crect^{\perp}_{\diamond})\ar[lu, "\text{Q-equiv}^{(\dagger)}"'] \\\\
\& \& \mathsf{L}_{\mathcal{S}_{\diamond}}\Qft(\Crect_{\diamond}^{\perp}) \ar[uu, "\text{Q-equiv}^{(\dagger)}" description]  \ar[rruu, "\text{Q-equiv}" description, bend right=10] \& \&
\end{tikzcd}
\end{equation}
The arrows marked with $(\dagger)$ are only known to be Quillen equivalences when $\Srect=\Srect_{\diamond}$ or when $(\Crect,\Crect_{\diamond})$ is either  $(\mathsf{COpen}(\Mrect),\mathsf{COpen}_{\diamond}(\Mrect))$ or $(\mathsf{CLoc}_2,\mathsf{CLoc}_{2,\diamond})$  and $\Srect$ is the set of Cauchy morphisms (see Examples \ref{ex_COpenMixingDiamonds} and \ref{ex_CLocMixingDiamonds}).

\end{section}

\begin{paragraph}{Acknowledgments}
The author would like to thank his advisors, R.\ Flores and F.\ Muro, for their support.  He also wants to thank M.\ Benini and A.\ Schenkel for their helpful comments and interest. C.\ Maestro also deserves recognition; his help with LaTeX matters is of unquestionable value. 
\end{paragraph}

\appendix
\begin{section}{Correction to: ``New model category structures for AQFT''}
\label{app_Corrigendum}
\begin{paragraph}{Abstract.} 
We observe that a technical claim made in \cite{carmona_new_2023} regarding left properness of operadic algebras is incorrect in the stated generality. We correct the affected results and provide counterexamples to the claim.
\end{paragraph}

\begin{subsection}{Introduction}
	The technical result \cite[Proposition 4.10]{carmona_new_2023} is incorrect in the stated generality. The issue arises at the end of its proof, where it was stated that an easy analysis of the pushouts appearing in \cite[Proposition 4.3.17]{white_bousfield_2018} yields the result. Such an analysis turned out to be quite non-trivial. It was explored in more detail and generality in \cite[\textsection 4.3]{carmona_enveloping_2024}, but the results in that work are not sufficient to recover the conclusion of \cite[Proposition 4.10]{carmona_new_2023}. This issue, being of technical nature, affects technically certain statements in \cite{carmona_new_2023}; specifically those where left Bousfield localization of left proper model categories was applied, i.e.\ \cite[Proposition 3.12 and Theorem 5.7/Corollary 5.8]{carmona_new_2023}. The solution we propose in this note is to use the more general left Bousfield localization for non-left proper model categories obtained in \cite{batanin_left_2024}. The output of such a construction is a relaxation of the notion of Quillen model structure, which is sufficient for practical purposes; for instance, consult  \cite{benini_equivalence_2024}. (As a consequence of this modification, any appearance of $\Lrect_{\mathcal{S}}\Qft^{(\Crect_{\diamond})}(\Crect^{\perp})$ in the text should state that it is a \textit{semimodel} structure/category. For instance, \cite[Theorems A and C]{carmona_new_2023}, and the summary \textsection6, must be slightly modified according to this.)
	
	In the sequel, we provide corrected statements and proofs for the affected results in \cite{carmona_new_2023}, which are presented here with the same numbering as in the original work and a superscript $\star$, and we discuss a counterexample invalidating \cite[Proposition 4.10]{carmona_new_2023} (see Propositions \ref{prop:counterexample lproper}, \ref{prop:counterexample lproperII}).
\end{subsection}

\begin{subsection}{Corrections}
	
	Recall that we work over the projective model structure on $\Ch_{\mathbb{k}}$, where $\mathbb{k}$ is a field of characteristic $0$. That is, we consider $\Ch_{\mathbb{k}}$-valued AQFTs over $\Crect^{\perp}$ and the associated projective model structure $\Qft(\Crect^{\perp})$. Also, let us comment that the semimodel structures we employ fall under \cite[Definition 2.1]{batanin_left_2024}, although various variations of this notion can be found in the literature.
	
	\begin{taggedprop}{$\textbf{3.12}^\star$}\label{prop 3.12}
		The left Bousfield localization of $\Qft(\Crect^{\perp})$ at $\mathcal{S}$ exists as a semimodel category, denoted $\mathsf{L}_{\mathcal{S}}\Qft(\Crect^{\perp})$, whose fibrant objects are fibrant objects in $\Qft(\Crect^{\perp})$ satisfying the homotopy $\Srect$-time-slice axiom.
	\end{taggedprop}
	\begin{proof} The existence of such a left Bousfield localization is guaranteed by \cite[Theorem A]{batanin_left_2024} since the projective model structure on operadic algebras valued in $\Ch_{\mathbb{k}}$ is tractable and combinatorial; see \cite[\textsection4.1]{benini_equivalence_2024}. The characterization of fibrant objects is the subject of \cite[Proposition 3.11]{carmona_new_2023}.
	\end{proof}
	
	\begin{rem} Notice that \cite[Theorem A]{batanin_left_2024} also equips $\mathsf{L}_{\mathcal{S}}\Qft(\Crect^{\perp})$ with the expected universal property in the realm of semimodel categories. Alternatively, one can construct Quillen adjunctions from or to $\mathsf{L}_{\mathcal{S}}\Qft(\Crect^{\perp})$ using \cite[Lemma 3.5]{carmona_simple_2024}. See for instance \cite[Proposition 4.7]{benini_equivalence_2024}.	Hence, the comparison results, e.g.\ \cite[Theorem 3.13 and Theorem 5.9]{carmona_new_2023}, remain unchanged.
	\end{rem}

	\begin{taggedcor}{$\textbf{5.8}^\star$} The semimodel structure  
		$\mathsf{L}_{\mathcal{S}}\Qft^{\Crect_{\diamond}}(\Crect^{\perp})$ presents the homotopy theory of algebraic quantum field theories which satisfy the homotopy $\Srect$-time-slice axiom and the $\Crect_{\diamond}$-local-to-global principle, i.e.\;the homotopy theory of  $\QFT_{\h\Srect}^{\Crect_{\diamond}}(\Crect^{\perp})$.
	\end{taggedcor}
	\begin{proof} The only modification with respect to \cite[Corollary 5.8]{carmona_new_2023} is the justification for the existence of $\mathsf{L}_{\mathcal{S}}\Qft^{\Crect_{\diamond}}(\Crect^{\perp})$ as a semimodel category. Similarly to Proposition \ref{prop 3.12}, this follows from \cite[Theorem A]{batanin_left_2024} since the extension model structure $\Qft^{\Crect_{\diamond}}(\Crect^{\perp})$ is combinatorial \cite[Proposition 4.3]{carmona_new_2023} and tractable by inspection; indeed, the generating (acyclic) cofibrations of $\Ch_{\mathbb{k}}^{(\times \ob(\Crect_{\diamond}))}$ have cofibrant domains, and thus so does any (semi)model structure obtained via left transfer from it. The computation of the bifibrant objects in this semimodel structure remains unchanged (see \cite[Theorem 5.7]{carmona_new_2023}).
	\end{proof}
	
\end{subsection}

\begin{subsection}{Counterexamples}
	Let us first show that the $\Ch_{\mathbb{k}}$-operad $\Op=\mathsf{Com}/(\upmu_3=0)$, obtained by killing the generator $\upmu_3$ of the $\mathbb{k}$-vector space $\mathsf{Com}(3)\cong \mathbb{k}$ in the \emph{non}-unital commutative operad $\mathsf{Com}$, does not satisfy the conclusion of \cite[Proposition 4.10]{carmona_new_2023}. This counterexample is motivated by a similar construction given in \cite{muro_correction_2017}. 
	
	For this task, we need to compute the enveloping operad $\Op_{\Aalg}$ for any $\Aalg\in \Alg_{\Op}(\Ch_{\mathbb{k}})$ (see \cite{carmona_enveloping_2024}, \cite{fresse_modules_2009}, \cite{harper_homotopy_2010} or \cite{white_bousfield_2018} for definitions). 
	
	\begin{lem}\label{lem: enveloping operad} Let $\Op=\faktor{\mathsf{Com}}{(\upmu_3=0)}$, where $\mathsf{Com}(3)=\mathbb{k}\upmu_3$, and $\Aalg\in \Alg_{\Op}(\Ch_{\mathbb{k}})$. Then,
		$$
		\Op_{\Aalg}(q)\cong
		\left\{
		\begin{array}{ll}
			\Aalg & \mbox{\qquad if $q=0$,} \\[3mm]
			\mathbb{k}\oplus \faktor{\Aalg}{\Aalg^2} & \mbox{\qquad if $q=1$,}\\[3mm]
			\mathbb{k} & \mbox{\qquad if $q=2$,}\\[3mm]
			0 & \mbox{\qquad if $q\geq 3$,}
		\end{array}
		\right.
		$$
		where $\Aalg^2$ is the image of the multiplication map $\upmu_{\Aalg}\colon\Aalg^{\otimes 2}\longrightarrow\Aalg$. 
	\end{lem}
	\begin{proof}
		We will use the explicit presentation of the enveloping operad from \cite[\textsection7.3]{harper_homotopy_2010}, i.e.\ 
		$$
		\begin{tikzcd}[ampersand replacement=\&]
			\Op_{\Aalg}(q)\cong \mathsf{colim}  \left(\bigoplus_{p\geq 0} \Op(p+q)\underset{\Upsigma_p}{\otimes}\Aalg^{\otimes p}\right. \ar[r, leftarrow, shift left=2,"\upd_0"] \ar[r, leftarrow, shift right=2,"\upd_1"'] \ar[r, lightgray] \& \left.\bigoplus_{p\geq 0} \Op(p+q)\underset{\Upsigma_p}{\otimes}(\Op\circ\Aalg)^{\otimes p}  \right)
		\end{tikzcd},
		$$
		together with the fact that $\Op(n)\cong \mathbb{k}$ if $n=1,2$ and $\Op(m)\cong 0$ otherwise.
		
		Let us divide the proof into cases:
		\begin{itemize}
			\item $q\geq 3$: one immediately observes $\Op_{\Aalg}(q)\cong 0$.  
			
			\item $q=2$: the only factor that survives in the coproduct correspond to $p=0$. This fact combined with inspection yields $\Op_{\Aalg}(2)\cong \mathbb{k}$.
			
			\item $q=1$: the factors that survive in the coproduct correspond to $p=0$ and $p=1$. The $p=0$ part contributes with a copy of $\mathbb{k}$ to $\Op_{\Aalg}(1)$. For the remaining part, letting $\overline{\Op_{\Aalg}(1)}$ be such that $\Op_{\Aalg}(1)=\mathbb{k}\oplus\overline{\Op_{\Aalg}(1)}$, we have
			$$
			\begin{tikzcd}[ampersand replacement=\&]
				\overline{\Op_{\Aalg}(1)}=  \mathsf{colim}  \left( \mathbb{k}\upmu_2\underset{\Upsigma_1}{\otimes}\Aalg^{\otimes 1}\right. \ar[r, leftarrow, shift left=2,"\upd_0"] \ar[r, leftarrow, shift right=2,"\upd_1"'] \ar[r, lightgray] \& \left. \mathbb{k}\upmu_2\underset{\Upsigma_1}{\otimes}\big(\mathbb{k}\upmu_1\otimes\Aalg\oplus \, \mathbb{k}\upmu_2\underset{\Upsigma_2}{\otimes}\Aalg^{\otimes 2}\big) \right)\\[-3mm]
				\cong  \mathsf{colim} \Big( \Aalg  \ar[r, leftarrow, shift left=2,"\upd_0=(\id_{\Aalg};\,0)"] \ar[r, leftarrow, shift right=2,"\upd_1=(\id_{\Aalg};\,\upmu_{\Aalg})"'] \ar[r, lightgray, "\upiota_1" description] \&  \Aalg\oplus\, (\Aalg^{\otimes 2})_{\Upsigma_2} \Big) \cong \faktor{\Aalg}{\Aalg^2}\;,  
			\end{tikzcd}
			$$ 
			where $\upiota_1$ denotes the inclusion into the first component.
			
			\item $q=0$: this case is automatic since the colimit computing $\Op_{\Aalg}(0)$ is the truncated bar-resolution of $\Aalg$; in symbols, $\Aalg\leftarrow \Op\circ\Aalg \leftleftarrows \Op\circ\Op\circ\Aalg$. 
		\end{itemize}
	\end{proof}

	\begin{prop}\label{prop:counterexample lproper} 
		Let $\Op=\faktor{\mathsf{Com}}{(\upmu_3=0)}$, where $\upmu_3$ is the generator of $\mathsf{Com}(3)$. Then, the projective model structure on the category of $\Op$-algebras   $\Alg_{\Op}(\Ch_{\mathbb{k}})$ is not left proper.
	\end{prop}
	\begin{proof} We will show that there is a weak equivalence $\upphi\colon\Aalg\to \Balg$ between $\Op$-algebras such that $\upphi\amalg \id \colon \Aalg\amalg \Op\circ\,\mathbb{k}\to \Balg\amalg \Op\circ\,\mathbb{k}$ is not a weak equivalence. This fact implies the claim since we have a cubical diagram in $\Alg_{\Op}(\Ch_{\mathbb{k}})$
		$$
		\begin{tikzcd}[ampersand replacement=\&]
			\Op\circ\,0\ar[dd, "!"'] \ar[rrrr, equal,  lightgray] \ar[rd, tail]
			\&[-3mm]\&[-3mm]\&[-3mm]\&[-3mm] \Op\circ\,0  \ar[dd,"!", near end] \ar[rd, tail]
			\\
			\&[-3mm] \Op\circ\,\mathbb{k} \ar[rrrr, equal, crossing over,  lightgray] \&[-3mm]\&[-3mm]\&[-3mm]\&[-3mm] \Op\circ\,\mathbb{k}\\
			\Aalg  \ar[rrrr,  lightgray, "\upphi" near end] \ar[rd, tail] \&[-3mm]\&[-3mm]\&[-3mm]\&[-3mm] \Balg \ar[rd, tail]
			\\
			\&[-3mm] \Aalg\amalg \Op\circ\,\mathbb{k} \ar[uu, leftarrow, crossing over] \ar[rrrr, crossing over, lightgray, "\upphi\amalg\id"'] \&[-3mm]\&[-3mm]\&[-3mm]\&[-3mm] \Aalg\amalg \Op\circ\,\mathbb{k} \ar[uu, leftarrow, crossing over]
		\end{tikzcd}\;\;.
		$$
		whose black faces are pushouts and because left properness can be characterized via cubes of this sort; see \cite[Proposition 15.4.4]{may_more_2012}.
		
		Combining \cite[Proposition 7.6]{harper_homotopy_2010} with Lemma \ref{lem: enveloping operad}, we obtain 
		$$
		\Aalg\amalg \Op\circ\,\mathbb{k}\cong \Op_{\Aalg}\circ\,\mathbb{k}=\bigoplus_{n\geq 0}\Op_{\Aalg}(n)_{\Upsigma_n}\cong \Aalg \oplus \Big(\mathbb{k}\oplus \faktor{\Aalg}{\Aalg^2}\Big)\oplus \mathbb{k}
		$$
		and thus, the map $\upphi\amalg \id \colon \Aalg\amalg \Op\circ\,\mathbb{k}\to \Balg\amalg \Op\circ\,\mathbb{k}$ is given in this case by
		$$
		\id_{\Aalg}\oplus \big(\id_{\mathbb{k}}\oplus\,[\upphi]\big)\oplus \id_{\mathbb{k}}\colon \Aalg \oplus \Big(\mathbb{k}\oplus \faktor{\Aalg}{\Aalg^2}\Big)\oplus \mathbb{k} \longrightarrow \Balg \oplus \Big(\mathbb{k}\oplus \faktor{\Balg}{\Balg^2}\Big)\oplus \mathbb{k}.
		$$
		In order for this map to be a weak equivalence of $\Op$-algebras, the factor $[\upphi]\colon \faktor{\Aalg}{\Aalg^2}\to \faktor{\Balg}{\Balg^2}$ should be a quasi-isomorphism. Let us construct a weak equivalence $\upphi\colon \Aalg\to \Balg$ so that $[\upphi]$ is not a quasi-isomorphism. 
		
		Consider the $\Op$-algebra $\Balg=\mathbb{k}y$ determined by $\vert y\vert=0$, $y^2=0$ and the $\Op$-algebra
		$$
		\Aalg=\big(\cdots 0 \xrightarrow{\;\;\phantom{\upd}\;\;} \mathbb{k}z \xrightarrow{\;\;\upd\;\;}\mathbb{k}x\oplus\, \mathbb{k}x^2\xrightarrow{\;\;\phantom{\upd}\;\;}0 \cdots\big)
		$$
		determined by $\vert z\vert=-1$, $zx=0=z^2$ and $x^2=\upd(z)$. The second $\Op$-algebra is a sort of resolution of the first one and so, the map $\upphi\colon \Aalg\to \Balg$ given by $z\mapsto 0$ and $x\mapsto y$ is clearly a weak equivalence. By direct computation,
		$$
		\faktor{\Balg}{\Balg^2}=\Balg \qquad \text{while} \qquad \faktor{\Aalg}{\Aalg^2}=\big(\cdots 0 \xrightarrow{\;\;\phantom{\upd}\;\;} \mathbb{k}z \xrightarrow{\;\;0\;\;}\mathbb{k}x\xrightarrow{\;\;\phantom{\upd}\;\;}0 \cdots\big)\cong \Balg\oplus\, \mathbb{k}[1].
		$$
		Therefore $[\upphi]\colon \faktor{\Aalg}{\Aalg^2}\longrightarrow \faktor{\Balg}{\Balg^2} $ cannot be a quasi-isomorphism.
	\end{proof} 
	
	The previous discussion is built on a $\mathbb{k}$-linear (or at least pointed) operad, and thus it is \textit{technically} not a counterexample of \cite[Proposition 4.10]{carmona_new_2023}. For this reason, the rest of this section is devoted to present a $\mathsf{Set}$-valued operad $\widetilde{\Op}$ satisfying the conclusion of Proposition \ref{prop:counterexample lproper}.  We are indebted to M.\ Markl and F.\ Muro for this counterexample.
	
	Let $\widetilde{\Op}=\faktor{\mathsf{Free}_{\mathsf{n}\Upsigma}(\{\widetilde{\upmu}\}_{(3)})}{(\widetilde{\upmu}\circ_i\widetilde{\upmu}=\widetilde{\upmu}\circ_j\widetilde{\upmu})_{1\leq i<j\leq 3}}$ be the non-symmetric $\mathsf{Set}$-operad generated by a totally associative 3-ary operation $\widetilde{\upmu}$.\footnote{Observe that  $\widetilde{\Op}$ is quadratic Koszul by \cite{markl_Koszulness} (note that $\widetilde{\Op}=t\mathcal{A}ss^3_0$ in the notation of loc.cit.), providing an even stronger statement against left properness of operadic algebras on $\Ch_{\mathbb{k}}$.} In particular, its underlying non-symmetric sequence (interpreted in $\Ch_{\mathbb{k}}$) is
	$$
	\widetilde{\Op}(n)\cong
	\left\{
	\begin{array}{ll}
		\mathbb{k}\widetilde{\upmu}^{i} & \mbox{\qquad if $n=2i+1$ with $i\geq 0$,} \\[3mm]
		0 & \mbox{\qquad otherwise,}
	\end{array}
	\right.
	$$
	where $\widetilde{\upmu}^0=\id$, $\widetilde{\upmu}^1=\widetilde{\upmu}$, and $\widetilde{\upmu}^{i+1}=\widetilde{\upmu}^i\circ_1\widetilde{\upmu}$.
	
	\begin{defn} Let $\Aalg\in \Alg_{\widetilde{\Op}}(\Ch_{\mathbb{k}})$ be an $\widetilde{\Op}$-algebra. Then, we define the complex 
		$$
		\begin{tikzcd}[ampersand replacement=\&]
			\mathcal{C}(\Aalg):= \mathsf{colim} \left(\Aalg^{\otimes 2}\right. \ar[r, leftarrow,  shift left=1,"\widetilde{\upmu}_{\Aalg}\otimes \,\id_{\Aalg}"] \ar[r, leftarrow,  shift right=1,"\id_{\Aalg}\otimes\, \widetilde{\upmu}_{\Aalg}"'] \&[8mm] \left. \Aalg^{\otimes 4}  \right)\cong \mathsf{coker}\big(\widetilde{\upmu}_{\Aalg}\otimes \id_{\Aalg}-\id_{\Aalg}\otimes \widetilde{\upmu}_{\Aalg}\big)
		\end{tikzcd}.
		$$
	\end{defn}
	
	\begin{lem}\label{lem:EnvelopingAlgebraOftAss3}
		Let $\widetilde{\Op}=\faktor{\mathsf{Free}_{\mathsf{n}\Upsigma}(\{\widetilde{\upmu}\}_{(3)})}{(\widetilde{\upmu}\circ_i\widetilde{\upmu}=\widetilde{\upmu}\circ_j\widetilde{\upmu})_{1\leq i<j\leq 3}}$ and $\Aalg\in \Alg_{\widetilde{\Op}}(\Ch_{\mathbb{k}})$. Then,
		\begin{align*}
			\widetilde{\Op}_{\Aalg}(1) 
			&\cong \mathbb{k}\oplus \big(\mathbb{k}\otimes\mathcal{C}(\Aalg)\big) \oplus \big(\mathcal{C}(\Aalg)\otimes \mathbb{k}\big)\oplus \mathcal{C}(\Aalg)^{\otimes 2}\oplus \Aalg^{\otimes 2}\\[2mm]
			&\cong \big(\mathbb{k}\oplus \mathcal{C}(\Aalg)\big)^{\otimes 2}\oplus \Aalg^{\otimes 2}.
		\end{align*}
	\end{lem}
	\begin{proof} As in Lemma \ref{lem: enveloping operad}, we use the explicit coequalizer description of the components of the enveloping operad associated to $(\widetilde{\Op},\Aalg)$. Notice that one has to make $\widetilde{\Op}$ a symmetric operad first to apply the same formulae. Doing so, the coequalizer becomes
		$$
		\begin{tikzcd}[ampersand replacement=\&]
			\displaystyle\widetilde{\Op}_{\Aalg}(1)\cong \mathsf{colim}  \left(\bigoplus_{\substack{p+q=2i\\p,q,i\geq 0}} \Aalg^{\otimes p}\otimes \Aalg^{\otimes q}\right. \ar[r, leftarrow, shift left=2,"\upd_0"] \ar[r, leftarrow, shift right=2,"\upd_1"'] \ar[r, lightgray] \&[-4mm] \displaystyle \left.\bigoplus_{\substack{p+q=2i\\p,q,i\geq 0,}} \big(\bigoplus_{j\geq 0}\Aalg^{\otimes 2j+1}\big)^{\otimes p} \otimes \big(\bigoplus_{l\geq 0}\Aalg^{\otimes 2l+1}\big)^{\otimes q}  \right)
		\end{tikzcd},
		$$
		where $\upd_0$ is induced by operadic composition,  $\widetilde{\upmu}^{a}\circ_k\widetilde{\upmu}^b=\widetilde{\upmu}^{a+b} $ for any $k$, and $\upd_1$ by the $\widetilde{\Op}$-algebra structure on $\Aalg$. Let us describe the canonical cocone from this diagram into our proposed expression for $\widetilde{\Op}_{\Aalg}(1)$:
		\begin{itemize}
			\item For $p=0$, $q=0$, we have the component $\mathbb{k}\cong \mathbb{k}\otimes\mathbb{k}$ which survives after taking the coequalizer (i.e.\ the projection is the identity into $\mathbb{k}$).
			\item For $p=0$, $q>0$, we must have $i>0$, and the projection becomes
			$$
			\Aalg^{\otimes 2i}=\Aalg^{\otimes 2(i-1)+1}\otimes \Aalg\xrightarrow{\quad\widetilde{\upmu}^{i-1}_{\Aalg}\otimes\,\id_{\Aalg} \quad} \Aalg\otimes \Aalg \xrightarrow{\quad \text{project}\quad}\mathcal{C}(\Aalg)=\mathbb{k}\otimes \mathcal{C}(\Aalg),
			$$
			where we can replace the first map with  $\id_{\Aalg}\otimes\, \widetilde{\upmu}^{i-1}_{\Aalg}$ by  definition of $\mathcal{C}(\Aalg)$. The symmetric case $p>0$, $q=0$ is analogous, and yields the component $\mathcal{C}(\Aalg)\otimes \mathbb{k}$. 
			\item For $p,q>0$ and $p$ even, we must have $q$ even, and the projection is
			$$
			\Aalg^{\otimes 2n}\otimes \Aalg^{\otimes 2m}\xrightarrow{\quad(\widetilde{\upmu}^{n-1}_{\Aalg}\otimes\,\id_{\Aalg})\otimes\,(\widetilde{\upmu}^{m-1}_{\Aalg}\otimes\, \id_{\Aalg}) \quad} \Aalg^{\otimes 2}\otimes \Aalg^{\otimes 2} \xrightarrow{\quad \text{project}^{\otimes 2}\quad}\mathcal{C}(\Aalg)\otimes \mathcal{C}(\Aalg),
			$$
			where the first map can be replaced by, e.g., $(\id_{\Aalg}\otimes\, \widetilde{\upmu}^{n-1}_{\Aalg})\otimes (\widetilde{\upmu}^{m-1}_{\Aalg}\otimes \id_{\Aalg})$ as above.
			\item For $p,q>0$ and $p$ odd, we must have $q$ odd, and the projection simply becomes
			$$
			\Aalg^{\otimes 2n+1}\otimes \Aalg^{\otimes 2m+1}\xrightarrow{\quad\widetilde{\upmu}^{n}_{\Aalg}\otimes\,\widetilde{\upmu}^{m}_{\Aalg} \quad} \Aalg\otimes \Aalg.
			$$
		\end{itemize}
		Inspection on the previous coequalizer shows that this is actually the universal cocone.
	\end{proof}
	
	
	\begin{prop}\label{prop:counterexample lproperII} 
		Let $\widetilde{\Op}=\faktor{\mathsf{Free}_{\mathsf{n}\Upsigma}(\{\widetilde{\upmu}\}_{(3)})}{(\widetilde{\upmu}\circ_i\widetilde{\upmu}=\widetilde{\upmu}\circ_j\widetilde{\upmu})_{1\leq i<j\leq 3}}$. Then, the projective model structure on the category of $\widetilde{\Op}$-algebras   $\Alg_{\widetilde{\Op}}(\Ch_{\mathbb{k}})$ is not left proper.
	\end{prop}
	\begin{proof}
		Following the same strategy as in Proposition \ref{prop:counterexample lproper}, we just need to construct a quasi-isomorphism of $\widetilde{\Op}$-algebras $\upphi\colon \Aalg\longrightarrow \Balg$ such that $\bigoplus_{n\geq 0}\widetilde{\Op}_{\Aalg}(n)_{\Upsigma_n}\longrightarrow \bigoplus_{n\geq 0}\widetilde{\Op}_{\Balg}(n)_{\Upsigma_n} $ fails to be a quasi-isomorphism (recall that $\widetilde{\Op}$ has to be symmetrized in order to use the same formulas as for $\Op$ in the cited result). Since we have just computed the unary operations of the enveloping operad associated to $\widetilde{\Op}$ in Lemma \ref{lem:EnvelopingAlgebraOftAss3}, we will content ourselves with checking that 
		$$
		\widetilde{\Op}_{\Aalg}(1)\cong \big(\mathbb{k}\oplus \mathcal{C}(\Aalg)\big)^{\otimes 2}\oplus \Aalg^{\otimes 2} \xrightarrow{\;\;(\id_{\mathbb{k}}\otimes\, \mathcal{C}(\upphi))^{\otimes 2}\otimes \,\upphi^{\otimes 2}\;\;} \big(\mathbb{k}\oplus \mathcal{C}(\Balg)\big)^{\otimes 2}\oplus \Balg^{\otimes 2}\cong \widetilde{\Op}_{\Balg}(1)
		$$
		is not a quasi-isomorphism, which is in any case enough for our purposes. In fact, we will show that $\mathcal{C}(\upphi)\colon \mathcal{C}(\Aalg)\longrightarrow \mathcal{C}(\Balg)$ is not a quasi-isomorphism.
		
		Consider the cochain complex
		$$
		\Aalg=\big(\cdots 0 \xrightarrow{\;\;\phantom{\upd}\;\;} \mathbb{k}z \xrightarrow{\;\;\upd\;\;}\mathbb{k}y\oplus\, \mathbb{k}y^3\xrightarrow{\;\;\phantom{\upd}\;\;}0 \cdots\big)
		$$
		with $\vert z\vert =-1$ and $\upd(z)=y^3$, and set the only non-trivial product to be $\widetilde{\upmu}(y,y,y)=y^3$. It is straightforward to check that $\Aalg$ is a well-defined $\widetilde{\Op}$-algebra. Let $\Balg=\Aalg/\langle z,y^3\rangle\cong \mathbb{k}y$ be the quotient $\widetilde{\Op}$-algebra with trivial $\widetilde{\Op}$-action. With these definitions, it is clear that the canonical projection $\Aalg\longrightarrow \Balg $ is a quasi-isomorphism. Let us compute the image under $\mathcal{C}$ of this map:
		\begin{itemize}
			\item Since $\widetilde{\upmu}_{\Balg}=0$, we obtain $\mathcal{C}(\Balg)=\Balg^{\otimes 2}\cong \mathbb{k}$.
			\item Denoting $\Aalg=\mathbb{k}\{y,y^3,z\}$, we have 
			$$
			\mathcal{C}(\Aalg)=\mathsf{coker}(\widetilde{\upmu}_{\Aalg}\otimes \id_{\Aalg}-\id_{\Aalg}\otimes\, \widetilde{\upmu}_{\Aalg})=\mathbb{k}\left\{\left[y\otimes y\right], \left[y\otimes y^3\right], \left[y\otimes z\right], \left[z\otimes y\right] , \left[z\otimes z\right]\right\},
			$$
			where the only non-vanishing differentials (of generators) are
			$
			\upd(\left[y\otimes z\right])=\left[y\otimes y^3\right]$ and $\upd(\left[z\otimes y\right])=\left[y\otimes y^3\right].
			$
			Thus,
			$$
			\mathsf{H}^*\mathcal{C}(\Aalg)\cong
			\left\{
			\begin{array}{ll}
				\mathbb{k}\left[y\otimes y\right] & \mbox{\qquad if $\ast=0$,} \\[3mm]
				\mathbb{k}\left[y\otimes z-z\otimes y\right] & \mbox{\qquad if $\ast=-1$,}\\[3mm]
				\mathbb{k}\left[z\otimes z\right] & \mbox{\qquad if $\ast=-2$,}\\[3mm]
				0 & \mbox{\qquad otherwise},
			\end{array}
			\right.
			$$
			and so $\mathcal{C}(\Aalg)\longrightarrow\mathcal{C}(\Balg)$ cannot be a quasi-isomorphism, concluding the claim.	
		\end{itemize}
	\end{proof}

	\begin{rem} P.\ Tamaroff suggested (in personal communication) that $\mathsf{Perm}$, the operad for (right) permutative algebras, might be another $\mathsf{Set}$-valued operad which is quadratic Koszul and yet doesn't have a left proper model category of algebras in $\Ch_{\mathbb{k}}$. Although we have not checked his claim, the strategy should be similar to the one presented above: one has to show that the functor $\Aalg\mapsto \mathsf{Perm}_{\Aalg}(1)$ does not preserve quasi-isomorphisms. 
	\end{rem}
	
	\begin{rem} The results in \cite[\textsection 4.3]{carmona_enveloping_2024} suggest the  still open possibility that the model structure $\Qft(\Crect^{\perp})$ could be left proper, although proving such a statement requires control over the enveloping operad $(\Op_{\Crect}^{\perp})_{\Aalg}$ for any $\Aalg\in \QFT(\Crect^{\perp})$. More concretely, we need to understand whether $(\Op_{\Crect}^{\perp})_{\Aalg}\to(\Op_{\Crect}^{\perp})_{\Balg}$ is a weak equivalence of $\Ch_{\mathbb{k}}$-operads when $\Aalg\to \Balg$ is a weak equivalence of AQFTs. The previous counterexamples build on the fact that this property is not automatically satisfied for all $\Ch_{\mathbb{k}}$-operads, and so whether $\Qft(\Crect^{\perp})$ is left proper or not may depend on the orthogonal category $\Crect^{\perp}$. For instance, when $\perp$ is the initial or terminal orthogonality relation on a non-empty category $\Crect$, $\Qft(\Crect^{\perp})$ \emph{is} left proper. 
	\end{rem}
	
\end{subsection}

\begin{paragraph}{Acknowledgments.}
The author would like to thank M.\ Markl, F.\ Muro, and P.\ Tamaroff for helpful and stimulating discussions. 
The author was supported by the Max-Planck Institute for Mathematics in the Sciences and partially supported by the project PID2024-157173NB-I00 funded by MCIN/AEI/10.13039/501100011033 and by FEDER, UE.
\end{paragraph}
\end{section}

\bibliography{Bibliography}

\end{document}